\documentclass[oneside,english]{amsart}
\usepackage[T1]{fontenc}
\usepackage[latin9]{inputenc}
\usepackage{geometry}
\usepackage{amsthm}
\usepackage{amsbsy}
\usepackage{amssymb}
\usepackage{graphicx}
\usepackage{esint}

\makeatletter
\numberwithin{equation}{section}
\numberwithin{figure}{section}
\theoremstyle{plain}
\newtheorem{thm}{\protect\theoremname}
  \theoremstyle{remark}
  \newtheorem{rem}[thm]{\protect\remarkname}
  \theoremstyle{plain}
  \newtheorem{cor}[thm]{\protect\corollaryname}

\renewcommand\[{\begin{equation}}
\renewcommand\]{\end{equation}} 

\makeatother

\usepackage{babel}
  \providecommand{\corollaryname}{Corollary}
  \providecommand{\remarkname}{Remark}
\providecommand{\theoremname}{Theorem}

\begin{document}

\title{A variational approach to modeling slow processes in stochastic dynamical
systems}

\author{Frank Noé and Feliks Nüske}
\begin{abstract}
The slow processes of metastable stochastic dynamical systems are
difficult to access by direct numerical simulation due the sampling
problem. Here, we suggest an approach for modeling the slow parts
of Markov processes by approximating the dominant eigenfunctions and
eigenvalues of the propagator. To this end, a variational principle
is derived that is based on the maximization of a Rayleigh coefficient.
It is shown that this Rayleigh coefficient can be estimated from statistical
observables that can be obtained from short distributed simulations
starting from different parts of state space. The approach forms a
basis for the development of adaptive and efficient computational
algorithms for simulating and analyzing metastable Markov processes
while avoiding the sampling problem. Since any stochastic process
with finite memory can be transformed into a Markov process, the approach
is applicable to a wide range of processes relevant for modeling complex
real-world phenomena.
\end{abstract}

\address{Department of Mathematics and Research Center \textsc{Matheon}\\
FU Berlin\\
Arnimallee 6, 14195 Berlin, Germany}

\email{frank.noe@fu-berlin.de, pycon@zedat.fu-berlin.de}

\maketitle

\section{Introduction}

In this article, we consider continuous-time Markov processes $\mathbf{z}_{t}\in\Omega$
living in a usually large state space $\Omega$ that is either continuous
or discrete. The process $\mathbf{z}_{t}$ is considered to be sufficiently
ergodic such that a unique stationary density (invariant measure)
$\mu$ exists. Independent of the details of the dynamics (such as
system size, potential, stochastic coupling, etc), there exists a
family of linear propagators $\mathcal{P}(\tau)$ which evolve the
probability density of states $\rho_{\tau}$ as:
\begin{equation}
\rho_{\tau}=\mathcal{P}(\tau)\:\rho_{0}\label{eq_propagator}
\end{equation}
Continuous-time Markov processes are useful models of real-world processes
in a variety of areas \cite{vanKampen}. Examples include macroscopic
phenomena including the evolution of financial and climate systems
\cite{McLeishKolkiewicz_DiffusionInFinance,Pelletier_Arxiv95_DiffusionClimate},
as well as microscopic dynamics such as the diffusion of cells in
liquids \cite{AnguigeSchmeiser_JMathBiol2009_CellDiffusion}, the
diffusion of biomolecules within cells \cite{SpaarHelms_BiophysJ06_Barstar},
the stochastic reaction dynamics of chemicals at surfaces \cite{RigerRogalReuter_PRL08_FirstPrinciplesKMC},
and the stochastic dynamics governing the structural dynamics of molecules
\cite{BestHummer_PNAS09_Diffusion}. Often, these dynamics are metastable,
i.e. they consist of slow processes between sets of state space that
have long lifetimes. In macromolecules, such slowly-exchanging sets
are called conformations, hence the union of the slowest dynamical
processes are there termed conformation dynamics \cite{SchuetteFischerHuisingaDeuflhard_JCompPhys151_146,Schuette_ICIAM07}. 

In practice, the slow dynamical processes are the ones which pose
the greatest difficulties to direct numerical simulation as they require
the longest simulation times. An extreme example is the atomistic
simulation of solvated biomolecules which would require the propagation
of a system with $10^{4}-10^{6}$ particles for $10^{7}-10^{10}$
time-steps, a task that is intractable or hardly tractable even with
special-purpose supercomputers \cite{Shaw_Science10_Anton}. However,
the slowest processes are also the ones that are the most interesting
in many systems. They often correspond to rare events that change
the global structure and/or the functional behavior of the system.
For example, in macromolecular systems, the slowest events often correspond
to functional conformational changes such as folding, binding or catalysis
\cite{NoeSchuetteReichWeikl_PNAS09_TPT,VoelzPande_JACS10_NTL9,BowmanVoelzPande_JACS11_FiveHelixBundle-TripletQuenching,HeldEtAl_BiophysJ10_AssociationTPT}.
Therefore, a method is sought that models the slow dynamical processes
of continuous-time Markov processes accurately and ideally in a way
that supports the efficient simulation of these processes without
the need to solve the full direct numerical simulation problem.

Especially in statistical physics, many theories and methods have
been proposed to model slow dynamical processes. Examples are rate
theories that describe the passage rate of a process across a surface
that separates metastable states \cite{Eyring_JCP35_TST,haenggi_RevModPhys62_251},
pathway-based theories and methods that describe the transition dynamics
of a system from a subset $A$ to a subset $B$ of state space \cite{EVandenEijnden_MMS04_Metastability,BolhuisChandlerDellagoGeissler_AnnuRevPhysChem02_TPS,MetznerSchuetteVandenEijnden_JCP06_TPT}
and network-based approaches that attempt to coarse-grain the high-dimensional
dynamics to a network of discrete jump events between sub-states or
landmarks \cite{Wales,WalesDoye_JCP119_12409}. These approaches usually
assume a separation of timescales between the slow and the fast processes
that results from vanishing smallness parameters (e.g. noise intensity,
temperature). In these cases, the mathematical analysis can be based
on large deviation estimates and variational principles\textbf{ \cite{FreidlinWentzell_Book_RandomPertubations,EVandenEijnden_MMS04_Metastability}}.

In this article, we consider a more general approach to describing
slow dynamical processes. When the operator $\mathcal{P}$ is compact
and self-adjoint, eq. (\ref{eq_propagator}) can be decomposed into
the propagator's spectral components 

\begin{eqnarray}
\rho_{\tau} & = & \mu+\sum_{i=2}^{m}a_{i}(\rho_{0})\lambda_{i}(\tau)l_{i}+\mathcal{P}_{\mathrm{fast}}(\tau)\:\rho_{0}\label{eq_propagator-decomposition}
\end{eqnarray}
where $\mbox{\ensuremath{\mu}},l_{2},l_{3},...$ are the propagator's
eigenfunctions and $\lambda_{i}(\tau)=\exp(-\kappa_{i}\tau)$ (sorted
in non-ascending order) are the propagator's real-valued eigenvalues
that decay exponentially in time with rates $\kappa_{i}$. $a_{i}(\rho_{0})$
are factors depending on the initial density $\rho_{0}$. We here
consider the situation that the eigenspaces of slow and fast processes
are orthogonal. For example, in the important case that the system
studied is a microscopic physical system in thermal equilibrium, such
that the process $\mathbf{z}_{t}$ is reversible with respect to the
invariant density $\mu$, the eigenspaces of slow and fast processes
are orthogonal for any choice of $m$. In such a case, $m$ is a model
parameter that can be chosen to distinguish the $m$ slow processes
of interest from the fast processes $\mathcal{P}_{\mathrm{fast}}(\tau)$
that are not of interest. For initial densities $\rho_{0}\in\mathrm{span}\{\mu,l_{2},...,l_{m}\}$,
the term for fast processes in Eq. (\ref{eq_propagator-decomposition})
vanishes and the propagation can be performed exactly using only the
dominant eigenfunctions. Note that Eq. (\ref{eq_propagator-decomposition})
represents the formulation of the dynamical process problem as a multi-scale
problem in time: $\mu$ represents the timescale $\infty$, $t_{2}=\kappa_{2}^{-1}=-\tau/\ln\lambda_{2}$
is the slowest dynamical process, etc. This formulation thus allows
essential characteristics of the dynamical process to be described
by treating its scales \emph{separately}. The relation between dominant
eigenvalues, exit times and rates, and metastable sets has been studied
by asymptotic expansions in certain smallness parameters as well as
by functional analytic means without any relation to smallness parameters
\cite{MaierStein_SIAP97_ExitDistribution,Deuflhard_LinAlgAppl_PCCA,Huisinga_AnnApplProbab04_PhaseTransitions,Bovier_JEMS04_Metastability1,Bovier_JEMS04_Metastability2}.
In particular, \cite{DaviesI,DaviesII} established fundamental relations
between the eigenvalues and -functions of $\mathcal{P}$ and metastable
sets.

The task is now to approximate the propagator's dominant eigenfunctions
$\mbox{\ensuremath{\mu}},l_{2},l_{3},...$ and eigenvalues $\lambda_{i}$.
In other disciplines, variational principles have been worked out
in order to approximate eigenfunctions and eigenvalues of known operators
such as a quantum-mechanical Hamiltonian \cite{SzaboOstlund}. In
contrast, for many complex dynamical systems, $\mathcal{P}(\tau)$
is either not known explicitly, or not available in a form that can
be transformed into Eq. (\ref{eq_propagator-decomposition}). Instead,
$\mathcal{P}(\tau)$ is given implicitly through stochastic realizations
of the process $\mathbf{z}_{t}$. Therefore, a variational principle
is sought that allows eigenfunctions $\mbox{\ensuremath{\mu}},l_{2},l_{3},...$
and eigenvalues $\lambda_{i}$ to be approximated through statistical
observables of $\mathbf{z}_{t}$. Such a variational principle will
be formulated in the present paper.

This problem has extensively been studied for specific functional
forms of the eigenfunctions $l_{i}$. When $l_{i}$ are approximated
via characteristic functions on a given set decomposition of state
space, i.e. $l_{i}\mu^{-1}\in\mathrm{span}\{\mathbf{1}_{S_{1}},...,\mathbf{1}_{S_{n}}\}$
with $S_{j}\subset\Omega$, the problem of finding the best approximation
to eigenfunctions and eigenvalues is solved by a Markov model or Markov
state model (MSM) \cite{SchuetteFischerHuisingaDeuflhard_JCompPhys151_146,Bovier_JEMS04_Metastability1,Deuflhard_LinAlgAppl_PCCA,SarichNoeSchuette_MMS09_MSMerror,NoeEtAl_PNAS11_Fingerprints,HuisingaSchmidt_06}.
In another version of MSMs, $l_{i}$ are approximated by committor
functions between a few pre-defined {}``cores'' that form a non-complete
subset of state space \cite{BucheteHummer_JPCB08,DjurdjevacSarichSchuette_MMS10_EigenvalueError,SchuetteEtAl_JCP11_Milestoning}.
Basis functions that form a partition of unity are used in \cite{Weber_PhD}.
Markov state models have been recently used a lot to model molecular
dynamics processes, especially in conjunction with large amounts of
distributedly simulated trajectories \cite{SwopePiteraSuits_JPCB108_6571,SinghalPande_JCP123_204909,ChoderaEtAl_JCP07,NoeHorenkeSchutteSmith_JCP07_Metastability,BucheteHummer_JPCB08,NoeSchuetteReichWeikl_PNAS09_TPT}.
Applications include conformational rearrangements and folding of
peptides, proteins and RNA \cite{ChoderaEtAl_JCP07,PanRoux_JCP08_MarkovModelPath,Bowman_JCP09_Villin,NoeSchuetteReichWeikl_PNAS09_TPT,VoelzPande_JACS10_NTL9,NoeEtAl_PNAS11_Fingerprints}.
In this application area, MSMs have had significant impact because
they can be estimated from relatively short simulation trajectories
and yet allow the system behavior to be predicted at long timescales. 

Despite this success, significant challenges remain. For example,
in most current applications, the discretization of state space is
done heuristically \emph{via} a Voronoi partition of state space obtained
from clustering available data points. The ability to construct a
state space discretization adaptively would tremendously aid the construction
of MSMs that are precise while avoiding the use of too many states.
Such an adaptive discretization must be guided by an objective function
that somehow measures the error made by the model. A bound to the
MSM discretization error has been derived in \cite{SarichNoeSchuette_MMS09_MSMerror}.
However, this error is not suitable to design a constructive discretization
approach as its evaluation requires knowledge of the exact eigenfunctions.
The variational principle derived in this paper only uses statistical
observables and is henceforth the basis for such a constructive approach
to adaptive discretization. Furthermore, it is shown that existing
Markov modeling approaches can be understood as a constrained optimal
solution of the variational principle \emph{via} a Ritz or Roothan-Hall
method with different choices of basis sets. Based upon this formulation,
further development of basis sets appropriate for describing complex
molecular conformation dynamics can be made.

The article is organized as follows: In Sec. \ref{sec_Theory} a variational
principle is formulated where the dominant eigenfunctions of stochastic
dynamical systems are approximated by maximizing a Rayleigh coefficient,
which - in the limit of the exact solution - is identical to the true
eigenvalues. This Rayleigh coefficient is linked to the computation
of correlation functions that can be evaluated without explicit knowledge
of the propagator $\mathcal{P}$. Sec. \ref{sec_Modeling} makes considerations
which types of functions may be practically useful to construct an
approximation to Eq. (\ref{eq_propagator-decomposition}) in complex
systems. Sec. \ref{sec_Example-1} shows numerical experiments on
a diffusion process in a one-dimensional double-well potential. Sec.
\ref{sec_Conclusion} concludes this study and makes suggestions for
the next steps\global\long\def\braketmuinv#1#2{\left\langle #1\mid#2\right\rangle _{\mu^{-1}}}

\section{Variational principle for conformation dynamics}

\label{sec_Theory}

\subsection{Basics}

Let $\Omega$ be a state space, and let us use $\mathbf{x,y}$ to
denote points in this state space. We consider a Markov process $\mathbf{z}_{t}$
on $\Omega$ which is stationary and ergodic with respect to its unique
stationary (invariant) distribution $\mu(\mathbf{x})\equiv p(\mathbf{z}_{t}=\mathbf{x})\:\forall t$.
The dynamics of the process $\mathbf{z}_{t}$ are characterized by
the transition density
\[
p(\mathbf{x},\mathbf{y};\,\tau)=p(\mathbf{z}_{t+\tau}=\mathbf{y}\mid\mathbf{z}_{t}=\mathbf{x}),
\]

which we assume to be independent of the time $t$. The correlation
density, i.e., the probability density of finding the process at points
$\mathbf{x}$ and $\mathbf{y}$ at a time spacing of $\tau$, is then
defined by
\[
C(\mathbf{x},\mathbf{y};\,\tau)=\mu(\mathbf{x})\, p(\mathbf{x},\mathbf{y};\,\tau)=p(\mathbf{z}_{t+\tau}=\mathbf{y},\:\mathbf{z}_{t}=\mathbf{x}).
\]

We further assume $\mathbf{z}_{t}$ to be reversible with respect
to its stationary distribution, i.e.:
\begin{eqnarray}
\mu(\mathbf{x})\, p(\mathbf{x},\mathbf{y};\,\tau) & = & \mu(\mathbf{y})\, p(\mathbf{y},\mathbf{x};\,\tau)\label{eq_reversibility1}\\
C(\mathbf{x},\mathbf{y};\,\tau) & = & C(\mathbf{y},\mathbf{x};\,\tau).\label{eq_reversibility2}
\end{eqnarray}

Reversibility is not strictly necessary but tremendously simplifies
the forthcoming expressions and their interpretation \cite{SarichNoeSchuette_MMS09_MSMerror}.
In physical simulations, reversibility is the consequence of the simulation
system being in thermal equilibrium with its environment, i.e. the
dynamics in the system is purely a consequence of thermal fluctuations
and there are no external driving forces.

If, at time $t=0$, the process is distributed according to a probability
distribution $\rho_{0}$, the corresponding distribution at time $\tau$
is given by:
\[
\rho_{\tau}(\mathbf{y})=\int_{\Omega}d\mathbf{x}\:\rho_{0}(\mathbf{x})\: p(\mathbf{x},\mathbf{y};\,\tau)=:\mathcal{P}(\tau)\rho_{0}.
\]

The time evolution of probability densities can be seen as the action
of a linear operator $\mathcal{P}(\tau)$, called the propagator of
the process. This is a well-defined operator on the Hilbert space
$L_{\mu^{-1}}^{2}(\Omega)$ of functions which are square-integrable
with respect to the weight function $\mu^{-1}$. The scalar-product
on this space is given by 
\[
\braketmuinv uv=\int_{\Omega}d\mathbf{x}u(\mathbf{x})v(\mathbf{x})\mu^{-1}(\mathbf{x}).
\]
If we assume the transition density to be a smooth and bounded function
of $\mathbf{x}$ and $\mathbf{y}$, the propagator can be shown to
be bounded, with operator norm less or equal to one. Since the stationarity
of $\mu$ implies $\mathcal{P}(\tau)\mu=\mu$, we even have $\|\mathcal{P}(\tau)\|=1$.
Reversibility allows us to show that the propagator is self-adjoint
and compact. Furthermore, using the definition of the transition density,
we can show that $\mathcal{P}(\tau)$ satisfies a Chapman-Kolmogorov
equation: For times $\tau_{1},\tau_{2}\geq0$, we have 
\[
\mathcal{P}(\tau_{1}+\tau_{2})=\mathcal{P}(\tau_{1})\mathcal{P}(\tau_{2}).
\]

\subsection{Spectral decomposition}

It follows from the above arguments that $\mathcal{P}(\tau)$ possesses
a sequence of real eigenvalues $\lambda_{i}(\tau)$, with $|\lambda_{i}(\tau)|\leq1$
and $|\lambda_{i}(\tau)|\rightarrow0$. Each of these eigenvalues
corresponds to an eigenfunction $l_{i}\in L_{\mu^{-1}}^{2}(\Omega)$.
The functions $l_{i}$ form an orthonormal basis of the Hilbert space
$L_{\mu^{-1}}^{2}(\Omega)$. Clearly, $\lambda_{1}(\tau)=1$ is an
eigenvalue with eigenfunction $l_{1}=\mu$. In many applications,
we can assume that $\lambda_{1}(\tau)$ is non-degenerate and $-1$
is not an eigenvalue. Additionally, there usually is a number $m$
of positive eigenvalues 
\[
1=\lambda_{1}(\tau)>\lambda_{2}(\tau)>\ldots>\lambda_{m}(\tau),
\]
which are separated from the remaining spectrum. Because of the Chapman-Kolmogorov
equation, each eigenvalue $\lambda_{i}(\tau)$ decays exponentially
in time, i.e. we have 
\[
\lambda_{i}(\tau)=\exp(-\kappa_{i}\tau)
\]
for some rate $\kappa_{i}\geq0$. Clearly, $\kappa_{1}=0$, $\kappa_{2},\ldots,\kappa_{m}$
are close to zero, and all remaining rates are significantly larger
than zero. If we now expand a function $u\in L_{\mu^{-1}}^{2}(\Omega)$
in terms of the functions $l_{i}$, i.e.

\[
u=\sum_{i=1}^{\infty}\braketmuinv u{l_{i}}l_{i},
\]

we can decompose the action of the operator $\mathcal{P}(\tau)$ into
its action on each of the basis functions:

\[
\begin{array}{ccc}
\mathcal{P}(\tau)u & = & \sum_{i=1}^{\infty}\braketmuinv u{l_{i}}\mathcal{P}(\tau)l_{i}\\
 & = & \sum_{i=1}^{\infty}\lambda_{i}(\tau)\braketmuinv u{l_{i}}l_{i}\\
 & = & \sum_{i=1}^{\infty}\exp(-\kappa_{i}\tau)\braketmuinv u{l_{i}}l_{i}.
\end{array}
\]

For lag times $\tau\gg\frac{1}{\kappa_{m+1}}$, all except the first
$m$ terms in the above sum have become very small \cite{SarichNoeSchuette_MMS09_MSMerror},
and to a good approximation we have

\[
\mathcal{P}(\tau)u\approx\sum_{i=1}^{m}\exp(-\kappa_{i}\tau)\braketmuinv u{l_{i}}l_{i}.
\]

Knowledge of the dominant eigenfunctions and eigenvalues is therefore
most helpful to the understanding of the process.
\begin{rem}
Instead of the propagator $\mathcal{P}(\tau)$, one can also consider
the transfer operator $\mathcal{T}(\tau)$, defined for functions
$u\in L_{\mu}^{2}(\Omega)$ by:

\[
\mathcal{T}(\tau)u(\mathbf{y})=\frac{1}{\mu(\mathbf{y})}\int_{\Omega}d\mathbf{x}p(\mathbf{x},\mathbf{y};\tau)\mu(\mathbf{x})u(\mathbf{x}).
\]

Using the unitary multiplication operator $\mathcal{M}:L_{\mu}^{2}(\Omega)\mapsto L_{\mu^{-1}}^{2}(\Omega)$,
defined by
\end{rem}
\[
\mathcal{M}u(\mathbf{x})=\mu(\mathbf{x})u(\mathbf{x}),
\]

we have

\[
\mathcal{P}(\tau)=\mathcal{MT}(\tau)\mathcal{M}^{-1},
\]

and consequently, the transfer operator inherits all of the above
properties from $\mathcal{P}(\tau)$. In particular, there is a sequence
of eigenfunctions 

\[
r_{i}=\mu^{-1}l_{i}
\]
of $\mathcal{T}(\tau)$, corresponding to the same eigenvalues $\lambda_{i}(\tau)$,
which are normalized w.r.t. to the scalar-product weighted with $\mu$
instead of $\mu^{-1}$. Especially, we have $r_{1}=\mu^{-1}\mu$=1
. The two operators can be treated as equivalent, and all of the above
could have been formulated in terms of $\mathcal{T}(\tau)$ as well.

\subsection{Rayleigh variational principle}

In nontrivial dynamical systems neither the correlation densities
$p(\mathbf{x},\mathbf{y};\,\tau)$ and $C(\mathbf{x},\mathbf{y};\,\tau)$
nor the eigenvalues $\lambda_{i}(\tau)$ and eigenfunctions $l_{i}$
are analytically available. This section provides a variational principle
based on which these quantities can be estimated from simulation data
generated by the dynamical process $\mathbf{z}_{t}$. For this, the
formalism introduced above is used to formulate the Rayleigh variational
principle used in quantum mechanics \cite{SzaboOstlund} for Markov
processes.

Let $f$ be a real-valued function of state, $f=f(\mathbf{x}):\Omega\rightarrow\mathbb{R}$.
Its autocorrelation with respect to the stochastic process $\mathbf{z}_{t}$
is given by:
\begin{eqnarray}
\mathrm{acf}(f;\,\tau) & = & \mathbb{E}[f(\mathbf{z}_{0})\: f(\mathbf{z}_{\tau})]=\int_{\mathbf{x}}\int_{\mathbf{y}}d\mathbf{x}\: d\mathbf{y}\: f(\mathbf{x})\: C(\mathbf{x},\mathbf{y};\,\tau)\: f(\mathbf{y})=\braketmuinv{\mathcal{P}(\tau)\mu f}{\mu f}.\label{eq_acf}
\end{eqnarray}
In the Dirac notation often used in physical literature, integrals
such as the one above may be abbreviated by $\mathbb{E}[f(\mathbf{x}_{0})\: f(\mathbf{x}_{\tau})]=\langle\mu f\mid\mathcal{\mathcal{P}}(\tau)\mid\mu f\rangle$. 
\begin{thm}
\emph{\label{thm_eigenvalue-exact}The autocorrelation function of
a weighted eigenfunction $r_{k}=\mu^{-1}l_{k}$ is its eigenvalue
$\lambda_{k}(\tau)$:} 
\[
\mathrm{acf}(r_{k};\,\tau)=\mathbb{E}\left[r_{k}(\mathbf{z}_{0})\: r_{k}(\mathbf{z}_{\tau})\right]=\lambda_{k}(\tau).
\]
\end{thm}
\begin{proof}
Using (\ref{eq_acf}) with $f=\mu^{-1}l_{k}$, it directly follows
that:

\[
\begin{array}{ccc}
\mathbf{\mathnormal{\mathrm{acf}(r_{k};\tau)}} & = & \braketmuinv{\mathcal{P}(\tau)l_{k}}{l_{k}}\\
 & = & \lambda_{k}(\tau)\braketmuinv{l_{k}}{l_{k}}\\
 & = & \lambda_{k}(\tau).
\end{array}
\]
\end{proof}
\begin{thm}
\emph{\label{thm_eigenvalue-bound-2}Let $\hat{l}_{2}$ be an approximate
model for the second eigenfunction, which is normalized and orthogonal
to the true first eigenfunction:
\begin{eqnarray}
\langle\hat{l}_{2},\mu\rangle_{\mu^{-1}} & = & 0\label{eq:-2}\\
\langle\hat{l}_{2},\hat{l}_{2}\rangle_{\mu^{-1}} & = & 1.\label{eq_normalization-second-efun}
\end{eqnarray}
}
\end{thm}
Then we find for $\hat{r}_{2}=\mu^{-1}\hat{l}_{2}$:\emph{
\[
\mathrm{acf}(\hat{r}_{2};\,\tau)=\mathbb{E}\left[\hat{r}_{2}(\mathbf{z}_{0})\:\hat{r}_{2}(\mathbf{z}_{\tau})\right]\le\lambda_{2}(\tau).
\]
}
\begin{proof}
The proof is an application of the Rayleigh variational method to
the operator $\mathcal{P}(\tau)$. If $\hat{l}_{2}$ is written in
terms of the basis of eigenfunctions $l_{i}$:
\[
\hat{l}_{2}=\sum_{i=2}^{\infty}a_{i}l_{i},
\]
where $a_{1}=0$ because of the orthogonality condition, we find:

\[
\begin{array}{ccc}
\mathrm{acf}(\hat{r_{2};}\tau) & = & \braketmuinv{\mathcal{P}(\tau)\hat{l_{2}}}{\hat{l_{2}}}\\
 & = & \sum_{i,j=2}^{\infty}a_{i}a_{j}\braketmuinv{\mathcal{P}(\tau)l_{i}}{l_{j}}\\
 & = & \sum_{i,j=2}^{\infty}a_{i}a_{j}\lambda_{i}(\tau)\braketmuinv{l_{i}}{l_{j}}\\
 & = & \sum_{i=2}^{\infty}a_{i}^{2}\lambda_{i}(\tau)\\
 & \leq & \lambda_{2}(\tau)\sum_{i=2}^{\infty}a_{i}^{2}\\
 & = & \lambda_{2}(\tau).
\end{array}
\]

The pre-last estimate is due to the ordering of the eigenvalues, and
the last equality results from the normalization condition \ref{eq_normalization-second-efun}
and Parseval's identity.\end{proof}
\begin{cor}
\emph{\label{thm_eigenvalue-bound-k}Similarly, let $\hat{l}_{k}$
be an approximate model for the $k$'th eigenfunction, with the normalization
and orthogonality constraints:
\begin{eqnarray}
\langle\hat{l}_{k},l_{i}\rangle_{\mu^{-1}} & = & 0,\:\:\:\:\:\forall i<k\label{eq:-14}\\
\langle\hat{l}_{k},\hat{l}_{k}\rangle_{\mu^{-1}} & = & 1,\nonumber 
\end{eqnarray}
}
\end{cor}
then\emph{
\[
\mathrm{acf}(\hat{r}_{k};\,\tau)=\mathbb{E}\left[\hat{r}_{k}(\mathbf{z}_{0})\:\hat{r}_{k}(\mathbf{z}_{\tau})\right]\le\lambda_{k}(\tau).
\]
}The proof is analogous to Theorem \ref{thm_eigenvalue-bound-2}.
\begin{rem}
Improved estimates than those of \ref{thm_eigenvalue-exact} to \ref{thm_eigenvalue-bound-k}
have been obtained by \cite{HuisingaSchmidt_06} for characteristic
functions. With a re-definition of the terminology, they can be directly
transferred to the case of mutually orthonormal basis functions. It
would be interesting to study the applicability of these results to
more general cases. However, obtaining such estimates is not the focus
of the present paper.
\end{rem}

\begin{rem}
The variational principle given by Theorems (\ref{thm_eigenvalue-exact})
to (\ref{thm_eigenvalue-bound-k}) is fulfilled for $\hat{l}_{k}$
with $k>2$ only if the $k-1$ dominant eigenfunctions are already
known.

In particular, the first eigenfunction, i.e. the stationary density
must be known. In practice, these eigenfunctions are approximated
via solving a variational principle. Nonetheless, some basic statements
can be made even if no eigenfunction is known exactly. For example,
it is trivial that when the estimated stationary density $\hat{\mu}$
is used in Theorem \ref{thm_eigenvalue-exact}, then the estimate
of the first eigenvalue is still always correctly 1:
\[
\mathrm{acf}(\hat{\mu}^{-1}\hat{\mu};\,\tau)=\mathrm{acf}(1;\,\tau)=1
\]
and from theorems \ref{thm_eigenvalue-exact} and \ref{thm_eigenvalue-bound-2}
it follows that any function $\hat{r}_{k}\neq\hat{\mu}$
\[
\mathrm{acf}(\hat{r}_{k};\,\tau)<1
\]
hence the eigenvalue 1 is simple and dominant also when estimating
eigenvalues from data. 
\end{rem}

\begin{rem}
An important insight at this point is that a variational principle
of conformation dynamics can be formulated in terms of correlation
functions. In contrast to quantum mechanics or other fields where
the variational principle has been successfully employed, no closed-form
expression of the operator $\mathcal{P}(\tau)$ is needed. The ability
to express the variational principle in terms of correlation functions
with respect to $\mathcal{P}(\tau)$ means that the eigenvalues to
be maximized can be directly estimated from simulation data. If statistically
sufficient realizations of $\mathbf{z}_{t}$ are available, then the
autocorrelation function can be estimated via:
\[
\mathrm{acf}(\hat{r}_{k};\,\tau)=\mathbb{E}(\hat{r}_{k}(\mathbf{z}_{0})\hat{r}_{k}(\mathbf{z}_{\tau}))\approx\frac{1}{N}\sum\hat{r}_{k}(\mathbf{z}_{0})\hat{r}_{k}(\mathbf{z}_{\tau}),
\]
where $N$ is the number of simulated time windows of length $\tau$.
We will try to use this in the application of the method.
\end{rem}

\subsection{Ritz method}

\label{sub_Ritz-method}The Ritz method is a systematic approach to
find the best possible approximation to the $m$ first eigenfunctions
of an operator simultaneously in terms of a linear combination of
orthonormal functions \cite{Ritz_JReineAngewMathe09_Variationsprobleme}.
Here the Ritz method is simply restated in terms of the present notation.
Let $\chi_{i}:\Omega\rightarrow\mathbb{R}$, $i\in\{1,...,m\}$ be
a set of $m$ orthonormal basis functions:
\[
\langle\chi_{i},\chi_{j}\rangle_{\mu^{-1}}=\delta_{ij},
\]
and let $\boldsymbol{\chi}$ denote the vector of these functions:
\[
\boldsymbol{\chi}(\mathbf{x})=[\chi_{1}(\mathbf{x}),...,\chi_{m}(\mathbf{x})]^{T}.
\]
We seek a coefficient matrix $\mathbf{B}\in\mathbb{R}^{m\times m}$,
\begin{equation}
\mathbf{B}=\left[\mathbf{b}_{1},...,\mathbf{b}_{m}\right]\label{eq_Ritz-coefficient-matrix}
\end{equation}
with the column vectors $\mathbf{b}_{i}=[b_{i1},...,b_{im}]^{T}$
that approximate the eigenfunctions of the propagator as:
\begin{eqnarray}
\hat{l}_{i}(\mathbf{x}) & = & \mathbf{b}_{i}^{T}\boldsymbol{\chi}(\mathbf{x})=\sum_{j=1}^{m}b_{ij}\chi_{j}(\mathbf{x}),\label{eq_Ritz-function-approximation}
\end{eqnarray}
with respect to the constraint that the functions $\hat{l}_{i}$ are
also normalized. It turns out that the solution $\mathbf{B}$ to the
eigenvalue equation 
\begin{equation}
\mathbf{H}\mathbf{B}=\mathbf{B}\boldsymbol{\hat{\Lambda}},\label{eq:ritz_matrix_eigenvalues}
\end{equation}
with individual eigenvalue/eigenvector pairs
\begin{equation}
\mathbf{H}\mathbf{b}_{i}=\mathbf{b}_{i}\hat{\lambda}_{i},\label{eq: ritz_individual_eigenvalues}
\end{equation}
and the density matrix $\mathbf{H}=[h_{ij}]$ defined by:
\begin{eqnarray}
h_{ij} & = & \int_{\mathbf{x}}\int_{\mathbf{y}}d\mathbf{x}\: d\mathbf{y}\mu^{-1}(\mathbf{x})\:\chi_{i}(\mathbf{x})\: C(\mathbf{x},\mathbf{y};\,\tau)\:\mu^{-1}(\mathbf{y})\chi_{j}(\mathbf{y})\label{eq_h_ij-integral}\\
 & = & \mathbb{E}[\mu^{-1}\chi_{i}(\mathbf{z}_{0})\:\mu^{-1}\chi_{j}(\mathbf{z}_{\tau})],\label{eq_h_ij-correlation}
\end{eqnarray}
yields the desired result. More precisely, the eigenvector $\mathbf{b}_{1}$
corresponding to the greatest eigenvalue $\hat{\lambda_{1}}$ from
\ref{eq: ritz_individual_eigenvalues} contains the coefficients of
the linear combination which maximizes the Rayleigh coefficient among
the functions $\chi_{i}$, and this maximum is given by $\hat{\lambda}_{1}.$
Consequently, $\hat{\lambda}_{1}$ should be as close as possible
to $\lambda_{1}=1$, and the function generated from $\mathbf{b}_{1}$
should model the stationary density $l_{1}$. But furthermore, the
remaining eigenvalues and eigenvectors generated from (\ref{eq:ritz_matrix_eigenvalues})
can be used as estimates of the other eigenvalues $\lambda_{2},\ldots,\lambda_{m}$:
\begin{cor}
The second estimated eigenvalue $\hat{\lambda}_{2}$ from (\ref{eq: ritz_individual_eigenvalues})
satisfies $\hat{\lambda}_{2}\leq\lambda_{2}$.\end{cor}
\begin{proof}
First of all, note that $\braketmuinv{\hat{l}_{2}}{\hat{l}_{1}}=0$
by the orthogonality of the eigenvectors of the matrix $\mathbf{H}$.
For the same reason, we find that:

\begin{eqnarray}
\braketmuinv{\mathcal{P}(\tau)\hat{l}_{1}}{\hat{l}_{2}} & = & \sum_{i,j=1}^{m}b_{1i}b_{2j}\braketmuinv{\mathcal{P}(\tau)\chi_{i}}{\chi_{j}}\nonumber \\
 & = & \sum_{i,j=1}^{m}b_{1i}b_{2j}h_{ij}\nonumber \\
 & = & \hat{\lambda}_{2}\sum_{i=1}^{m}b_{1i}b_{2i}\nonumber \\
 & = & 0.\label{eq:rayleigh_coeff_l1_l2}
\end{eqnarray}

Now, let $\hat{l}=x\hat{l}_{1}+y\hat{l}_{2}$ be a linear combination
of the first two model eigenfunctions which is normalized such that

\begin{eqnarray}
1 & = & \braketmuinv{\hat{l}}{\hat{l}}\nonumber \\
 & = & x^{2}\braketmuinv{\hat{l}_{1}}{\hat{l}_{1}}+y^{2}\braketmuinv{\hat{l}_{2}}{\hat{l}_{2}}\nonumber \\
 & = & x^{2}+y^{2}.\label{eq:normalization_linear_comb_l1_l2}
\end{eqnarray}

Using (\ref{eq:rayleigh_coeff_l1_l2}) and (\ref{eq:normalization_linear_comb_l1_l2}),
computing the Rayleigh coefficient of $\hat{l}$ results in:

\begin{eqnarray}
\braketmuinv{\mathcal{P}(\tau)\hat{l}}{\hat{l}} & = & x^{2}\braketmuinv{\mathcal{P}(\tau)\hat{l}_{1}}{\hat{l}_{1}}+y^{2}\braketmuinv{\mathcal{P}(\tau)\hat{l}_{2}}{\hat{l}_{2}}\nonumber \\
 & = & x^{2}\hat{\lambda}_{1}+y^{2}\hat{\lambda}_{2}\nonumber \\
 & = & \hat{\lambda}_{1}-y^{2}(\hat{\lambda}_{1}-\hat{\lambda}_{2}).\label{eq:estimated_second_ev_bound_below}
\end{eqnarray}

which is bounded from below by $\hat{\lambda}_{2}$. Clearly, there
is a normalized linear combination $\hat{l}$ of $\hat{l}_{1}$ and
$\hat{l}_{2}$ which is orthogonal to $l_{1}$. By (\ref{eq:estimated_second_ev_bound_below})
and the variational principle, we conclude that:

\begin{eqnarray*}
\hat{\lambda}_{2} & \leq & \braketmuinv{\mathcal{P}(\tau)\hat{l}}{\hat{l}}\\
 & \leq & \lambda_{2}.
\end{eqnarray*}
\end{proof}
\begin{rem}
Due to the equality between Eq. (\ref{eq_h_ij-integral}) and (\ref{eq_h_ij-correlation})
the elements of the $\mathbf{H}$ matrix can be estimated as correlation
functions of a simulation of the process $\mathbf{z}_{t}$, as mentioned
above, provided that a sufficient approximation of $\mu$ is at hand.

\end{rem}

\subsection{Roothaan-Hall method}

\label{sub_Rothaan-Hall-method-1}The Roothaan-Hall method is a generalization
of the Ritz method used for solving the linear parameter optimization
problem for the case when the basis set is not orthogonal \cite{Roothaan_RevModPhys51_RoothaanHall,Hall_ProcRoySocLonA51_RoothaanHall}.
Let the matrix $\mathbf{S}\in\mathbb{R}^{m\times m}$ with elements
\[
S_{ij}=\langle\chi_{i},\chi_{j}\rangle_{\mu^{-1}}
\]

be the matrix of overlap integrals with the normalization conditions
$S_{ii}=1$. Note that $\mathbf{S}$ has full rank if and only if
all $\chi_{i}$ are pairwise linearly independent. The optimal solutions
$\mathbf{b}_{i}$ in the sense of Eqs (\ref{eq_Ritz-coefficient-matrix})-(\ref{eq_Ritz-function-approximation})
are found by the eigenvectors of the generalized eigenvalue problem:
\begin{equation}
\mathbf{H}\mathbf{B}=\mathbf{S}\mathbf{B}\hat{\mathbf{\Lambda}}\label{eq:roothan_hall_gen_ev}
\end{equation}

with the individual eigenvalue/eigenvector pairs:
\[
\mathbf{H}\mathbf{b}_{i}=\mathbf{S}\mathbf{b}_{i}\hat{\lambda}_{i}.
\]

\begin{rem}
The Ritz and Roothaan-Hall methods are useful for eigenfunction models
that are expressed in terms linear combinations of basis functions.
Non-linear parameter models can also be handled with nonlinear optimization
methods. In such nonlinear cases it needs to be tested whether there
is a unique optimum or not.
\end{rem}

\subsection{Markov state model\label{sub:Markov-state-model}}

As an example, let $\{S_{1},...,S_{n}\}$ be pairwise disjoint sets
partitioning $\Omega$ and let $\pi_{i}=\int_{S_{i}}d\mathbf{x}\mu(\mathbf{x})$
be the stationary probability of set $S_{i}\subset\Omega$. Consider
the piecewise constant functions
\[
\chi_{i}=\frac{1}{\sqrt{\pi_{i}}}\mathbf{1}_{S_{i}}
\]
where $\mathbf{1}_{S_{i}}$ is the characteristic function that is
1 for $\mathbf{x}\in S_{i}$ and 0 elsewhere. Since $S_{i}\cap S_{j}=\emptyset$
for all $i\neq j$ these functions form a basis set with $\langle\chi_{i},\chi_{j}\rangle_{\mu}=\delta_{ij}$.
Therefore, we can directly use them as a model for the transfer operator
eigenfunctions $r_{k}$. Evaluation of the corresponding $\mathbf{H}$
matrix yields:

\begin{eqnarray}
h_{ij} & = & \frac{1}{\sqrt{\pi_{i}\pi_{j}}}\int_{\mathbf{x}}\int_{\mathbf{y}}d\mathbf{x}\: d\mathbf{y}\:\mathbf{1}_{S_{i}}\: C(\mathbf{x},\mathbf{y};\,\tau)\:\mathbf{1}_{S_{i}}\label{eq:-3}\\
 & = & \frac{1}{\sqrt{\pi_{i}\pi_{j}}}\int_{S_{i}}\int_{S_{j}}d\mathbf{x}\: d\mathbf{y}\: C(\mathbf{x},\mathbf{y};\,\tau)\nonumber \\
 & = & \frac{c_{ij}}{\sqrt{\pi_{i}\pi_{j}}}\nonumber \\
 & = & T_{ij}\sqrt{\frac{\pi_{i}}{\pi_{j}}}\nonumber 
\end{eqnarray}
where $c_{ij}=\mathbb{P}(\mathbf{z}_{t+\tau}\in S_{j},\mathbf{z}_{t}\in S_{i})$
is the joint probability of observing the process in sets $S_{i}$
and $S_{j}$ with a time lag of $\tau$ while $T_{ij}=\mathbb{P}(\mathbf{z}_{t+\tau}\in S_{j}\mid\mathbf{z}_{t}\in S_{i})$
is the corresponding transition probability. Thus, computing the optimal
step-function approximation to the true eigenfunctions $r_{i}=\mu^{-1}l_{i}$
and eigenvalues $\lambda_{i}(\tau)$ via the Ritz method is the same
as computing eigenvalues and eigenvectors of the Markov model transition
matrix $\mathbf{T}=[T_{ij}]$ and scaling them appropriately. This
conclusion can also be obtained from Ref. \cite{SarichNoeSchuette_MMS09_MSMerror}
\emph{via} a different route.

\section{Modeling}

\label{sec_Modeling}Section \ref{sec_Theory} has provided a general
variational principle for approximating the dominant eigenvalues and
eigenfunctions of Markov processes. In order to apply this principle
to complex systems, a useful level of modeling the eigenfunctions
in terms of basis functions needs to be found, and appropriate classes
of basis functions must be identified. This sections attempt a first
approach to this problem by making general considerations for what
modeling schemes might be appropriate.

\subsection{Half-weighted eigenfunctions}

Is it beneficial to directly model the propagator eigenfunctions $l_{k}$,
their weighted counterparts $r_{k}=\mu^{-1}l_{k}$ or rather yet another
set of functions? We would like to use a model that has the following
properties:
\begin{enumerate}
\item As basis functions $\chi_{i}$ it is preferred to use local functions,
i.e. either functions with compact support, or at least with the property
$\lim_{|\mathbf{x}|\rightarrow\infty}\chi_{i}(\mathbf{x})\rightarrow0$.
Such locality is useful to direct the computation effort to specific
regions of state space and may aid the adaptive refinement of the
eigenfunction approximation by specifically adding basis functions
that add local refinements. Since we also aim at modeling eigenfunctions
as linear combinations of basis functions we cannot use the $r_{k}$
eigenfunctions that are not local.
\item We would like to be able to pre-compute as many expressions as possible
analytically. When using appropriate basis functions $\chi_{i}$ and
$\chi_{j}$ it may be possible to calculate analytic solutions of
the integrals $\langle\chi_{i},\chi_{j}\rangle$, albeit this feature
is usually destroyed when weighting with the stationary density as
in $\langle\chi_{i},\chi_{j}\rangle_{\mu^{-1}}$. Therefore we will
also avoid using the eigenfunctions $l_{k}$ that would require such
a weighting.
\end{enumerate}
Consider a rewrite of Eq. (\ref{eq_acf}) as:
\begin{eqnarray}
\mathrm{acf}(r_{k};\,\tau) & = & \int_{\mathbf{x}}\int_{\mathbf{y}}d\mathbf{x}\: d\mathbf{y}\:\mu^{-\frac{1}{2}}(\mathbf{x})\: l_{k}(\mathbf{x})\:\mu^{-\frac{1}{2}}(\mathbf{x})\: C(\mathbf{x},\mathbf{y};\,\tau)\:\mu^{-\frac{1}{2}}(\mathbf{y})\:\mu^{-\frac{1}{2}}(\mathbf{y})\: l_{k}(\mathbf{y})\nonumber \\
 & = & \int_{\mathbf{x}}\int_{\mathbf{y}}d\mathbf{x}\: d\mathbf{y}\:\phi_{k}(\mathbf{x})\: S(\mathbf{x},\mathbf{y};\,\tau)\:\phi_{k}(\mathbf{y})\label{eq_acf-half-weighted}
\end{eqnarray}

where the {}``half-weighted'' eigenfunctions:
\[
\phi_{i}(\mathbf{x})=\frac{l_{i}(\mathbf{x})}{\mu^{\frac{1}{2}}(\mathbf{x})},
\]
and the {}``half-weighted'' correlation density:
\[
S(\mathbf{x},\mathbf{y};\,\tau)=\frac{C(\mathbf{x},\mathbf{y};\,\tau)}{\mu^{\frac{1}{2}}(\mathbf{x})\mu^{\frac{1}{2}}(\mathbf{y})}
\]
have been defined. When now modeling the half-weighted eigenfunctions
$\phi_{i}$ using some basis set, the following nice properties are
obtained:
\begin{enumerate}
\item Local basis functions can be used. This follows from $\lim_{|\mathbf{x}|\rightarrow\infty}\phi_{i}(\mathbf{x})=\lim_{|\mathbf{x}|\rightarrow\infty}\mu^{\frac{1}{2}}(\mathbf{x})r_{i}(\mathbf{x})\rightarrow0$
\item The normalization condition requires a non-weighted scalar product:
\begin{eqnarray}
\langle l_{i},l_{j}\rangle_{\mu^{-1}} & = & \left\langle \frac{l_{i}}{\mu^{1/2}},\frac{l_{j}}{\mu^{1/2}}\right\rangle \label{eq:-4}\\
 & = & \langle\phi_{i},\phi_{j}\rangle=\delta_{ij}.\nonumber 
\end{eqnarray}

\item When $\langle\chi_{i},\chi_{j}\rangle$ is analytically computable
and $\phi_{k}=\sum_{i}c_{i}\chi_{i}$, then $\langle\phi_{k},\phi_{k}\rangle$
is also analytically computable.
\item The first half-weighted eigenfunction has eigenvalue 1 and is identical
to the half-weighted stationary density
\[
\phi_{1}(\mathbf{x})=\frac{l_{1}(\mathbf{x})}{\mu(\mathbf{x})^{1/2}}=\mu(\mathbf{x})^{1/2}.
\]

\item When models of $\mu(\mathbf{x})^{1/2}$ and $\phi_{k}$ are available,
the Rayleigh coefficient in Eq. (\ref{eq_acf-half-weighted}) can
be estimated numerically as the autocorrelation of $\frac{\phi_{i}}{\mu^{1/2}}$
\item When defining a propagator $\mathcal{P}^{*}$ in half-weighted space
via:
\begin{eqnarray}
p_{\tau}(\mathbf{y}) & = & \mathcal{P}(\tau)\circ p_{0}(\mathbf{x})\label{eq:-5}\\
p_{\tau}(\mathbf{y}) & = & \int_{\mathbf{x}}d\mathbf{x}\: p(\mathbf{x})\: p(\mathbf{x},\mathbf{y};\,\tau)\nonumber \\
\frac{p_{\tau}(\mathbf{y})}{\mu(\mathbf{y})^{1/2}} & = & \int_{\mathbf{x}}d\mathbf{x}\:\frac{p(\mathbf{x})}{\mu(\mathbf{x})^{1/2}}\:\frac{\mu(\mathbf{x})^{1/2}}{\mu(\mathbf{y})^{1/2}}\: p(\mathbf{x},\mathbf{y};\,\tau)\nonumber \\
p_{\tau}^{*}(\mathbf{y}) & = & \int_{\mathbf{x}}d\mathbf{x}\: p_{0}^{*}(\mathbf{x})\: p^{*}(\mathbf{x},\mathbf{y};\,\tau)\nonumber \\
 & = & \mathcal{P}^{*}(\tau)\circ p_{0}^{*}(\mathbf{x})\nonumber 
\end{eqnarray}
then $\mathcal{P}^{*}$ is self-adjoint and has orthogonal eigenfunctions
\[
\phi_{i}\lambda_{i}=\mathcal{P}^{*}(\tau)\phi_{i}
\]

\end{enumerate}
It follows from theorem (\ref{thm_eigenvalue-exact}) that the exact
eigenvalues are calculated by the Rayleigh coefficients of the exact
eigenfunctions:
\begin{eqnarray}
\lambda_{i} & = & \left\langle \mu^{-\frac{1}{2}}\phi_{i}\mid\mathcal{C}\mid\mu^{-\frac{1}{2}}\phi_{i}\right\rangle =\mathrm{acf}(\mu^{-\frac{1}{2}}\phi_{i};\,\tau)\label{eq:-6}
\end{eqnarray}
while they are approximated from below by the Rayleigh coefficients
of the approximate eigenfunctions:
\begin{eqnarray}
\lambda_{i}\ge\hat{\lambda}_{i} & = & \left\langle \mu^{-\frac{1}{2}}\hat{\phi}_{i}\mid\mathcal{C}\mid\mu^{-\frac{1}{2}}\hat{\phi}_{i}\right\rangle =\mathrm{acf}(\mu^{-\frac{1}{2}}\hat{\phi}_{i};\,\tau).\label{eq:-7}
\end{eqnarray}
This Rayleigh-coefficient can be directly sampled: for a given trajectory
$\mathbf{z}_{t}$, it can be estimated as:
\begin{eqnarray}
\lambda_{i}\ge\hat{\lambda}_{i} & = & \mathbb{E}_{t}\left[\hat{\mu}(\mathbf{z}_{t})^{-\frac{1}{2}}\hat{\phi}_{i}(\mathbf{z}_{t})\hat{\mu}(\mathbf{z}_{t+\tau})^{-\frac{1}{2}}\hat{\phi}_{i}(\mathbf{z}_{t+\tau})\right]\label{eq:-8}
\end{eqnarray}
Thus, for a given trajectory $\mathbf{z}_{t}$, the optimal eigenfunctions
$\hat{\phi}_{i}$ can be calculated by maximizing the Rayleigh coefficient,
using e.g. the Ritz or the Roothaan-Hall method.

\subsection{Gaussian Basis functions}

\label{sub_Gaussian-Basis}

In complex dynamical processes such as molecular dynamics of biomolecules,
one has to devise basis sets that can be evaluated in high dimensions.
While the present work provides merely a starting point for identifying
appropriate basis sets that go beyond common choices such as the step
function basis or the committor basis, we here suggest a possible
choice that is potentially applicable to the molecular dynamics setting.
Empirically, it has been found that the stationary densities of biomolecules
in the essential subspace is often clustered \cite{Grubmueller_PhysRevE52_2893}.
Therefore, we put forward the idea that $\mu(\mathbf{x})^{1/2}$ and
the other half-weighted eigenfunctions can be approximated by a Gaussian
mixture. 

Let us thus assume that the state space $\Omega$ is a metric space
with distance $d(\mathbf{x},\mathbf{y})$, and let us model the half-weighted
invariant density $\hat{\mu}(\mathbf{x})^{1/2}$ in terms of Gaussian
basis functions:

\[
\hat{\mu}(\mathbf{x})^{1/2}=\sum_{i}a_{i}\exp\left(-\frac{d(\mathbf{x},\mathbf{y}_{i})}{2\sigma^{2}}\right)
\]
where $a_{i}\in\mathbb{R}$, $\mathbf{y}_{i}\in\Omega$ and $\sigma\in\mathbb{R}$
are amplitude, mean and shape parameters. The invariant density can
then be analytically given:
\begin{eqnarray}
\hat{\mu}(x)=\left(\hat{\mu}(x)^{1/2}\right)^{2} & = & \left(\sum_{i}a_{i}\exp\left(-\frac{d(\mathbf{x},\mathbf{y}_{i})}{2\sigma^{2}}\right)\right)^{2}\label{eq:-9}\\
 & = & \sum_{i}a_{i}^{2}\exp\left(-\frac{d(\mathbf{x},\mathbf{y}_{i})}{\sigma^{2}}\right)+\sum_{i<j}2a_{i}a_{j}\exp\left(-\frac{d(\mathbf{x},\mathbf{y}_{i})+d(\mathbf{x},\mathbf{y}_{j})}{2\sigma^{2}}\right).\nonumber 
\end{eqnarray}

Furthermore, consider that the half-weighted eigenfunctions be given
in terms of the same Gaussian basis:
\[
\hat{\phi}_{k}(\mathbf{x})=\sum_{i}b_{ki}\exp\left(-\frac{d(\mathbf{x},\mathbf{y}_{i})}{2\sigma^{2}}\right)
\]

where $b_{ki}$ must be appropriately chosen to guarantee orthogonality
with respect to the invariant density. The corresponding unweighted
eigenfunctions $\hat{r}_{k}$ are
\[
\hat{r}_{k}(x)=\frac{\hat{\phi}_{k}(x)}{\hat{\phi}_{1}(x)}=\frac{\sum_{i}b_{ki}\exp\left(-\frac{d(\mathbf{x},\mathbf{y}_{i})}{2\sigma^{2}}\right)}{\sum_{i}a_{i}\exp\left(-\frac{d(\mathbf{x},\mathbf{y}_{i})}{2\sigma^{2}}\right)},
\]

which does not have a simple form, but can be evaluated point-wise.
In order to enforce the normalization $\langle\hat{\phi}_{k},\hat{\phi}_{l}\rangle=\delta_{kl}$
we consider
\begin{eqnarray}
\langle\hat{\phi}_{k},\hat{\phi}_{l}\rangle & = & \int d\mathbf{x}\sum_{i}\sum_{j}b_{ki}b_{lj}\exp\left(-\frac{d(\mathbf{x},\mathbf{y}_{i})+d(\mathbf{x},\mathbf{y}_{j})}{2\sigma^{2}}\right),\label{eq:-10}
\end{eqnarray}

which can be analytically evaluated when $\Omega$ is an Euclidean
space. 

Using Gaussian Ansatz functions in half-weighted space may thus have
important practical benefits. However, other basis sets, especially
sets of other Radial basis functions than Gaussians may also be a
good choice for high-dimensional systems that deserve further investigation.

\section{Numerical examples}

\label{sec_Example-1}

\subsection{Metastable potential from a Gaussian stationary density}

The example is chosen such that it is tractable by direct grid discretization
so as to be able to generate a reference solution. Different optimization
methods for the variational problem and choices of basis sets are
considered and illustrated.

Let $\Omega=\mathbb{R}$ be our state space with points $x\in\Omega$.
First, a {}``Gaussian hat'' function is defined \emph{via}:
\[
gh(x;a,s):=\exp\left(-\frac{(x-a)^{2}}{2s^{2}}\right)
\]

where $a\in\mathbb{R}$ is the mean and $s\in\mathbb{R}$ the standard
deviation. We define a stationary density from two Gaussians:
\begin{eqnarray*}
\mu(x) & := & \frac{1}{2\sqrt{\pi}}\left(gh(x,-2,1)+gh(x,2,1)\right)
\end{eqnarray*}

The corresponding dimensionless generating potential is given by
\[
U(x)=-\ln\left(\mu(x)\right)
\]

which exerts a force on a particle at position $x$ of:
\[
f(x)=-\nabla U(x)=\frac{1}{\mu(x)}\frac{d}{dx}\mu(x).
\]

Using Smoluchowski dynamics and Euler discretization, a time-step
$x\overset{\tau}{\rightarrow}y$ is given by
\begin{equation}
y=x+\tau f(x)+\sqrt{2\tau}\eta\label{eq:smoluchowski_time_step}
\end{equation}

where $\eta$ is a normally distributed random variable (white noise).
The transition density can hence be written as:
\[
p(x,y;\tau)=\mathcal{N}_{y}(x+\tau f(x),\sqrt{2\tau})
\]

and the correlation density is given by:
\[
C(x,y;\tau)=\mu(x)\: p(x,y;\tau)
\]

Now we are concerned about estimation of eigenfunctions. The first
half-weighted eigenfunction is the square root of the stationary density:
\[
\phi_{1}=\sqrt{\mu(x)}
\]

such that $\phi_{1}^{2}(x)=\mu(x)$. We here assume $\mu(x)$ to be
known, although in practice it must be estimated.

\begin{figure}
\begin{centering}
(a)\includegraphics[width=0.3\columnwidth]{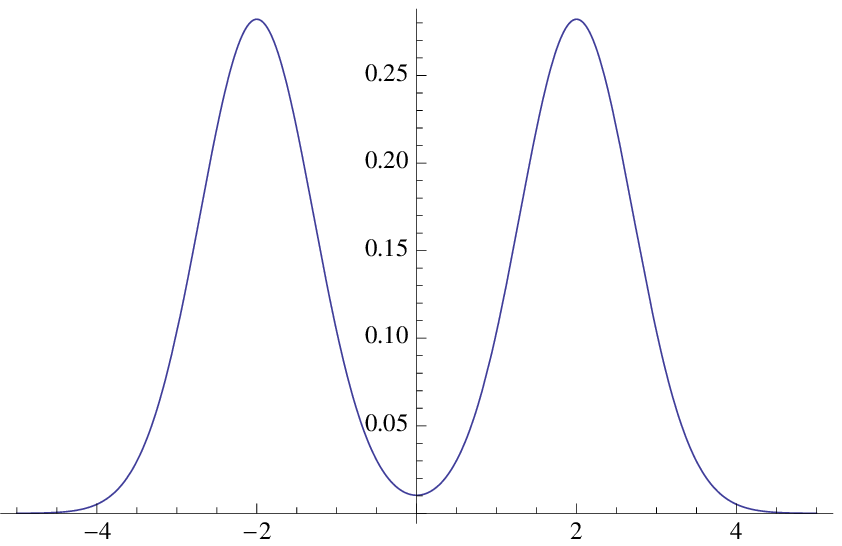} (b)\includegraphics[width=0.3\columnwidth]{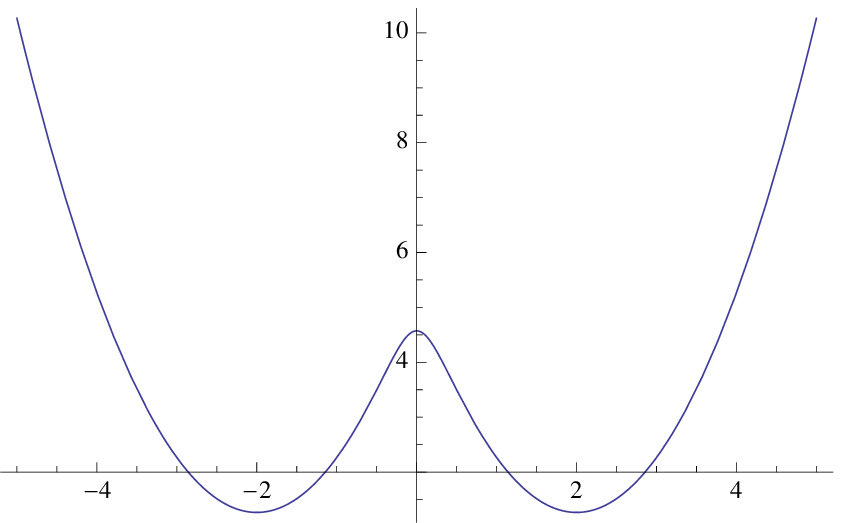}
\par\end{centering}

\begin{centering}
(c)\includegraphics[width=0.3\columnwidth]{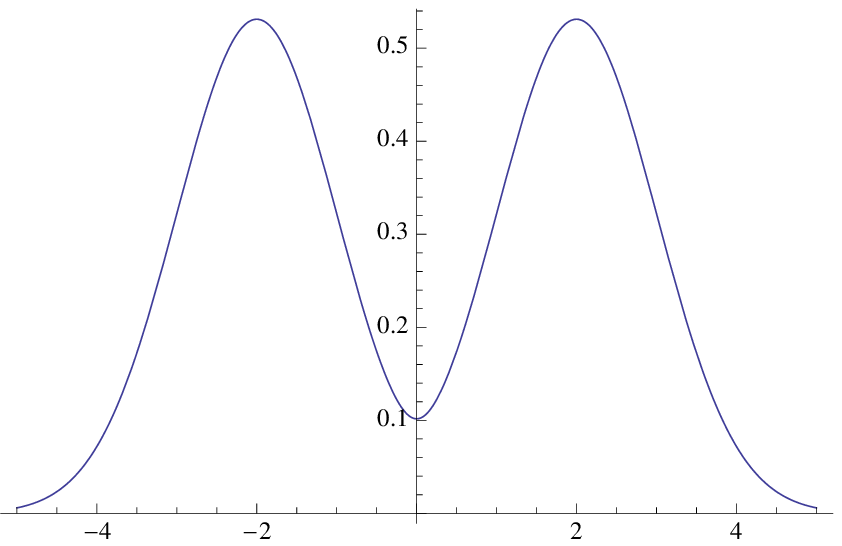} (d)\includegraphics[width=0.3\columnwidth]{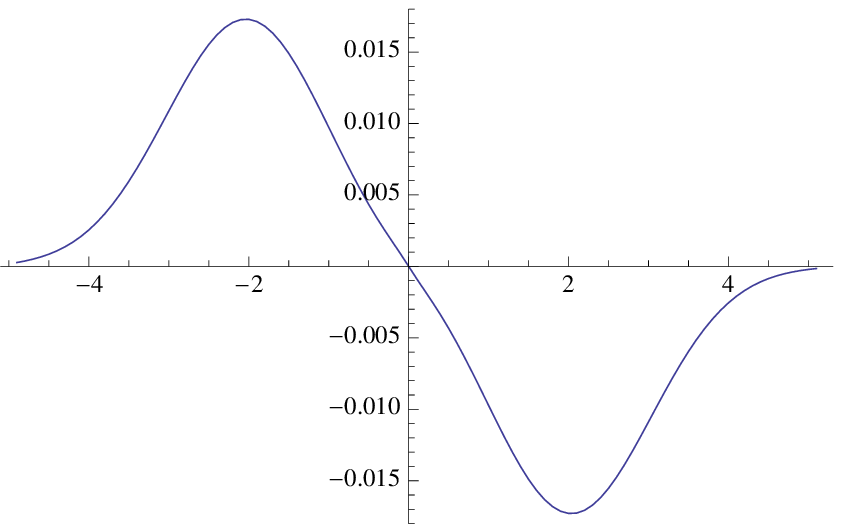} 
\par\end{centering}

\caption{\label{fig:2-well-1}Two-well potential with Smoluchowski dynamics,
$\tau=0.025$. (a) Double-Gaussian density $\mu(x)$, (b) the corresponding
potential $U(x)$, (c) the half-weighted density $\phi_{1}(x)=\sqrt{\mu(x)}$,
(d) slowest-process eigenfunction $\phi_{2}(x)$ from direct numerical
solution.}
\end{figure}

\subsection{Ritz method with characteristic functions (Markov state model)}

We aim at approximating the eigenvalues and eigenfunctions of the
true propagator via the Ritz method described in Sec. \ref{sub_Ritz-method}.
Here, a basis set $\boldsymbol{\chi}=(\chi_{1}(x),...,\chi_{N}(x))^{T}$
consisting of $N=$20 characteristic functions in the range $x\in[-6,6]$
defined by:
\[
\chi_{i}=\frac{1}{\sqrt{\pi_{i}}}\mathbf{1}_{[-6+0.6i,\:-5.4+0.6i]}
\]

where $\mathbf{1}_{[a,b]}$ is the characteristic function that is
1 on the interval $[a,b]$ and 0 outside, and $\pi_{i}=\int_{s_{i}}\mu(x)dx$
is the stationary probability of the set $S_{i}=[-6+0.6i,\:-5.4+0.6i]$.
The corresponding density-matrix $\mathbf{H}=[h_{ij}]\in\mathbb{R}^{N\times N}$
defined by (Dirac notation):
\[
h_{ij}=\langle\chi_{i}\mid\mathcal{C}\mid\chi_{j}\rangle
\]

takes the form
\[
h_{ij}=\frac{\sqrt{\pi_{i}}\int_{S_{i}}dx\:\int_{S_{j}}dy\: C(x,y;\:\tau)}{\sqrt{\pi_{j}}}
\]
as in Sec. \ref{sub:Markov-state-model}.

The $\mathbf{H}$ matrix was calculated by direct numerical integration
using Mathematica and the eigenvalue problem was subsequently solved,
yielding the optimal coefficient vectors $\mathbf{c}_{1}$ and $\mathbf{c}_{2}$
that provide the approximations $\hat{\phi}_{1}\approx\phi_{1}$,
$\hat{\phi}_{2}\approx\phi_{2}$ and $\hat{\lambda}_{1}\approx\lambda_{2}$,
$\hat{\lambda}_{2}\approx\lambda_{2}$. The results are given in Fig.
\ref{fig:2-well-Ritz-1} a and b, indicating that the eigenvalues
are approximated to two significant digits while the eigenfunctions
retain a significant discretization error that arises from the step-function
basis used.

\begin{figure}

\begin{centering}
a)\includegraphics[width=0.3\columnwidth]{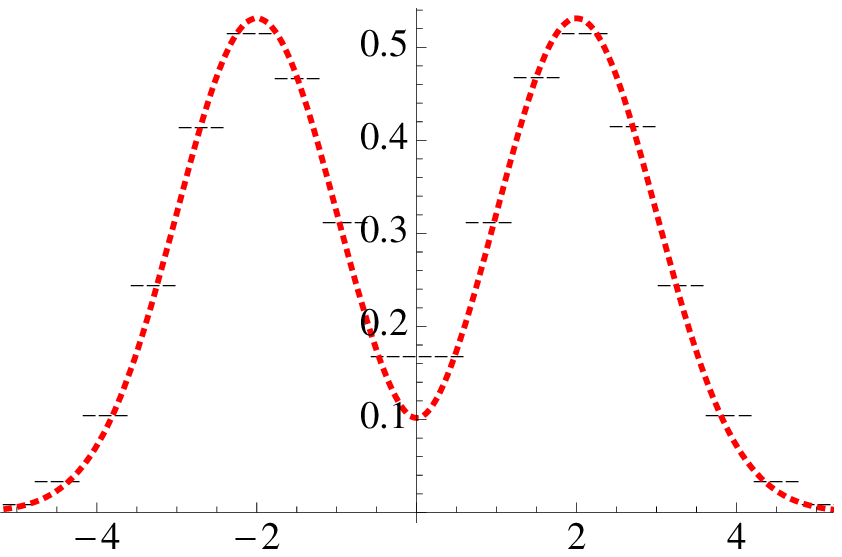} b)\includegraphics[width=0.3\columnwidth]{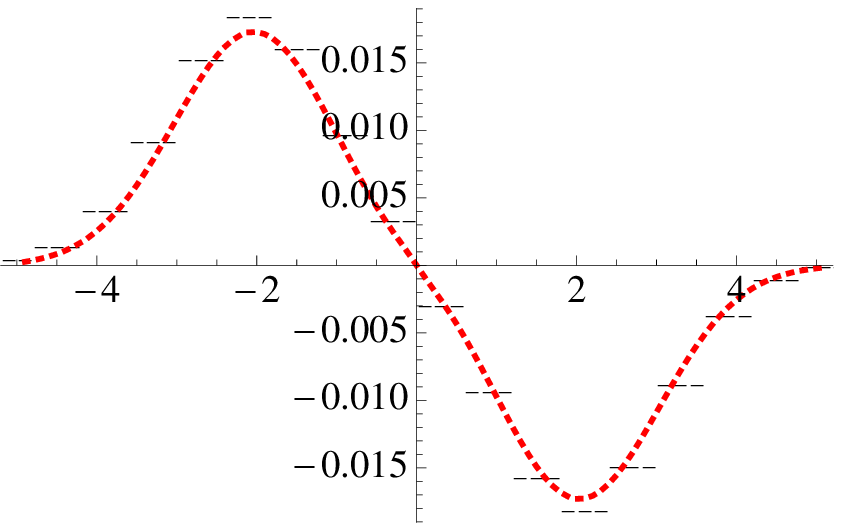}
\par\end{centering}

\begin{centering}
c)\includegraphics[width=0.3\columnwidth]{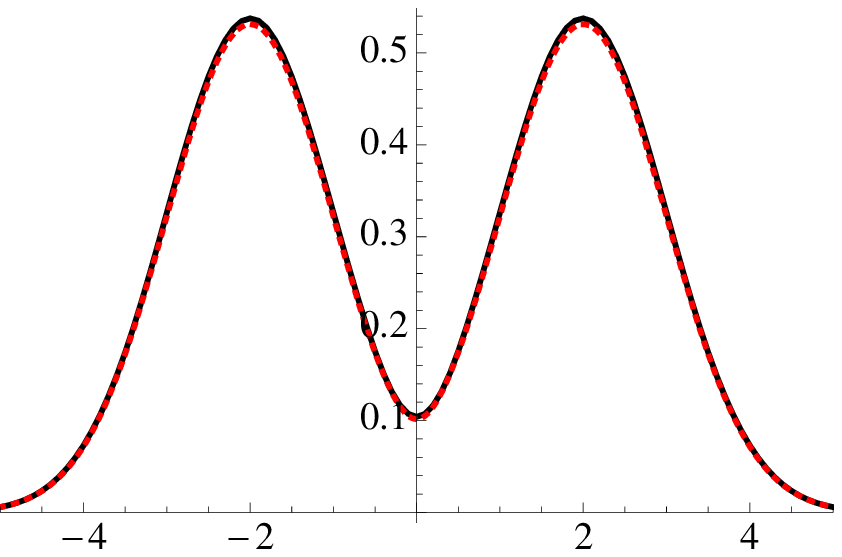} d)\includegraphics[width=0.3\columnwidth]{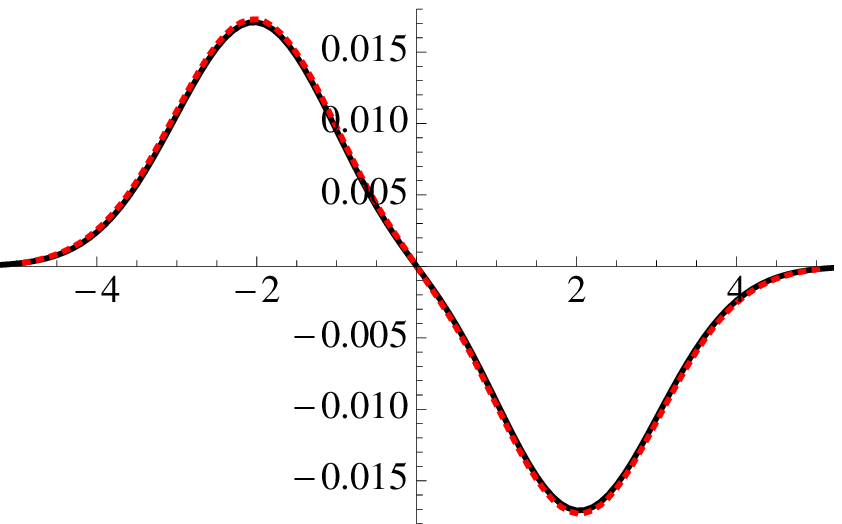}
\par\end{centering}

\begin{centering}
e)\includegraphics[width=0.3\columnwidth]{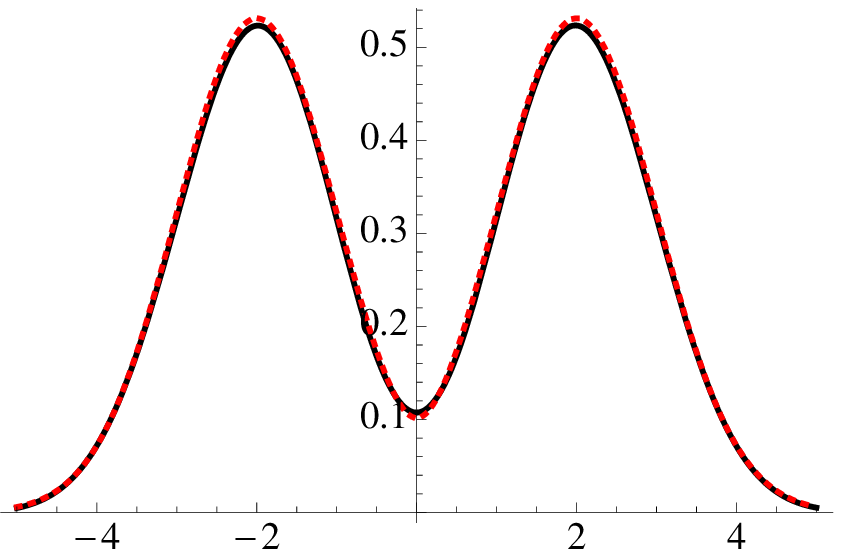} f)\includegraphics[width=0.3\columnwidth]{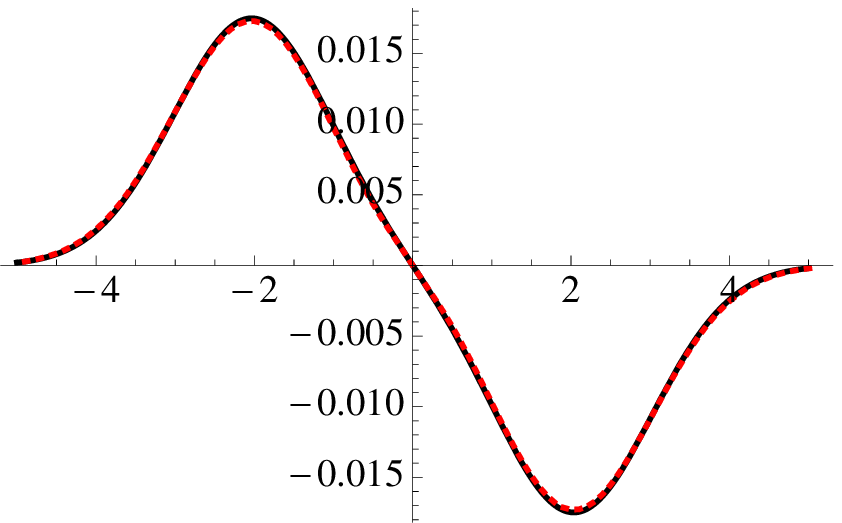}
\par\end{centering}

\caption{\label{fig:2-well-Ritz-1}Approximation to the eigenfunction shown
in Fig. \ref{fig:2-well-1} using different methods and basis sets.
The true eigenvalues are 1.0, 0.998913. The reference solutions are
shown in red dotted lines while the approximations are shown in black
solid lines. a,b) MSM / Ritz method with 20 characteristic functions
in the range $x\in[-6,6]$. Eigenvalues 1.0, 0.980384. c,d) Ritz
method with a basis set of 20 Hermite functions, Eigenvalues 1.0,
0.998913. e,f) Roothaan-Hall method with a basis set of 11 Gaussians
Eigenvalues 1.0, 0.995507.}
\end{figure}

\subsection{Ritz method with a Hermite basis}

In order to arrive at a smooth solution we employ the Ritz method
with a smooth orthogonal function basis. Here, we choose the Hermite
functions, defined by:
\[
\psi_{i}(x)=(-1)^{i}(2^{i}i!\sqrt{\pi})^{-1/2}\mathrm{e}^{x^{2}/2}\frac{d^{i}}{dx^{i}}\mathrm{e}^{-x^{2}}.
\]
The Hermite functions are local ($\lim_{|\mathbf{x}|\rightarrow\infty}\chi_{i}(\mathbf{x})\rightarrow0$)
and are thus useful to model the behavior of the eigenfunctions $\phi_{k}$
where the stationary density is significantly larger than zero. The
basis functions are defined to be the normalized Hermite functions
with
\[
\chi_{i}=\frac{\psi_{i}}{\sqrt{\langle\psi_{i},\psi_{i}\rangle}}.
\]
such that $\langle\chi_{i},\chi_{j}\rangle=\delta_{ij}$. 

Using the basis set $\boldsymbol{\chi}=[\chi_{0},...,\chi_{19}]^{T}$,
the $\mathbf{H}$ matrix was calculated by direct numerical integration
using Mathematica and the Ritz method was used to approximate eigenvalues
and eigenfunctions of the propagator. As shown in Fig. \ref{fig:2-well-Ritz-1}c
and d, a nearly perfect approximation of both eigenvalues and eigenfunctions
is obtained even though the number of basis functions used is identical
to the MSM approach of the previous section. However, the MSM approach
has the advantage that it can be employed in high-dimensional spaces
which is not the case with Hermite basis functions.

\subsection{Roothaan-Hall method with a Gaussian basis}

In order to have a hope to solve high-dimensional problems, one must
resort to simple basis functions, ideally ones with analytical properties
that can be practically evaluated in high-dimensional spaces. Therefore,
we here suggest the use of Gaussian basis functions as described in
Sec. \ref{sub_Gaussian-Basis}. In the one-dimensional case, the Gaussians
used are:
\[
gh_{i}(x)=\exp\left(-\frac{(x-y_{i})^{2}}{2\sigma^{2}}\right).
\]
Here, we use $\sigma=1$ and $y_{i}=(-5,-4,...,4,5)$. Gaussian basis
functions are not orthogonal. We therefore calculate the overlap matrix
$\mathbf{S}=[s_{ij}]$ with
\begin{eqnarray}
s_{ij} & = & \langle gh_{i},gh_{j}\rangle\label{eq:-11-1}\\
 & = & \int_{-\infty}^{\infty}dx\:\exp\left(-\frac{(x-y_{i})^{2}+(x-y_{j})^{2}}{2\sigma^{2}}\right)\nonumber 
\end{eqnarray}
that can be evaluated analytically. The $\mathbf{H}=[h_{ab}]$ matrix
is again defined by (Dirac notation):
\[
h_{ab}=\langle gh_{a}\mid\mathcal{C}\mid gh_{b}\rangle.
\]

Using the Roothaan-Hall method (Sec. \ref{sub_Rothaan-Hall-method-1}),
the best approximation to the propagator eigenfunctions $\hat{\phi}_{i}=\langle\mathbf{b}_{i},\boldsymbol{\chi}\rangle$
are found by the eigenvectors of the generalized eigenvalue problem
\[
\mathbf{S}^{-1}\mathbf{H}\mathbf{b}_{i}=\hat{\lambda}_{i}\mathbf{b}_{i}.
\]
As shown in Fig. \ref{fig:2-well-Ritz-1} c and d, an also nearly
perfect approximation of both eigenvalues and eigenfunctions is achieved
even though the basis is smaller than the previous basis sets. Since
the Gaussian basis set is a good candidate for being used in high-dimensional
spaces, this is probably the most useful result so far. Note that
the matrix inversion $\mathbf{S}^{-1}$ can be efficiently calculated
with sparse matrix methods when a cutoff is used to set nearly-non-overlapping
pairs with $h_{ab}\approx0$ to 0.

\subsection{Nonlinear optimization}

The previous methods used exclusively linear combinations of basis
functions. A greater degree of freedom in approximating the propagator
eigenfunctions is achieved by using additional shape parameters in
the basis functions. This, however, leads to a nonlinear optimization
problem that is in general difficult to solve and may have multiple
optima. However, we briefly illustrate the approach on our one-dimensional
example. We make the Ansatz for the second half-weighted eigenfunction:
\begin{eqnarray}
\hat{\phi}_{2}(x) & = & \frac{1}{Z}\left(-gh(x,y_{2},s_{2})+gh(x,y_{2},s_{2})\right)\label{eq:-12-1}\\
Z & = & \left[\int_{-\infty}^{\infty}dx\:[-gh(x,y_{2},s_{2})+gh(x,y_{2},s_{2})]\right]^{1/2}\nonumber 
\end{eqnarray}

The normalization constant makes sure that $\langle\hat{\phi}_{2},\hat{\phi}_{2}\rangle=1$.
The constraint $\langle\hat{\phi}_{2},\phi_{1}\rangle=0$ is here
ensured by the fact that $\phi_{1}$ is an even function and $\hat{\phi}_{2}$
is an odd function.

The optimal parameters $\hat{y}_{2}$ and $\hat{s}_{2}$ are found
by maximizing the Rayleigh coefficient:
\[
(\hat{y}_{2},\hat{s}_{2})=\arg\max_{y_{2},s_{2}}\left\langle \frac{\hat{\phi}_{2}(y_{2},s_{2})}{\mu^{1/2}}\mid\mathcal{C}\mid\frac{\hat{\phi}_{2}(y_{2},s_{2})}{\mu^{1/2}}\right\rangle .
\]

Fig. \ref{fig_Nonlinear-optimization-1} shows the results of varying
$y_{2}$ and $s_{2}$ as well as the local optimum for $y_{2}=1$
and $s_{2}=0.8$. In this case, a good approximation to the eigenvector
could be achieved with a 2-term ansatz function. However, the general
usefulness of the nonlinear approach for high-dimensional problems
remains to be evaluated.

\begin{figure}
\begin{centering}
(a)\includegraphics[width=0.3\columnwidth]{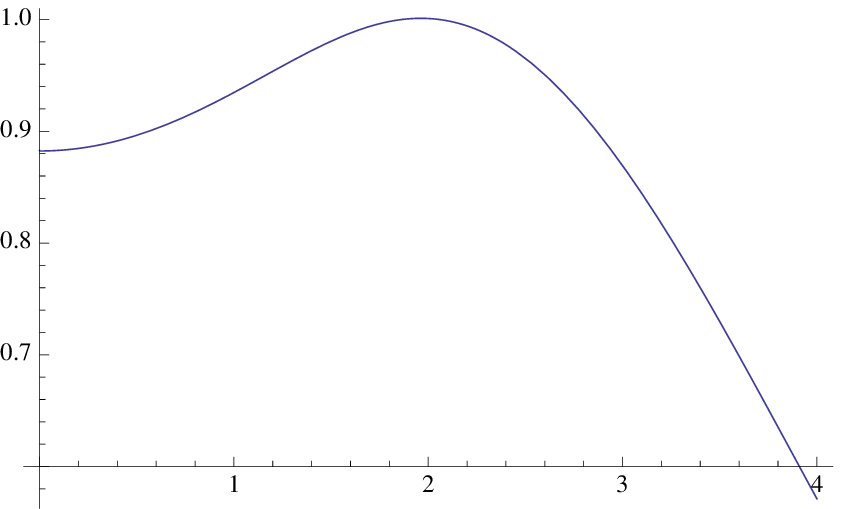} (b)\includegraphics[width=0.3\columnwidth]{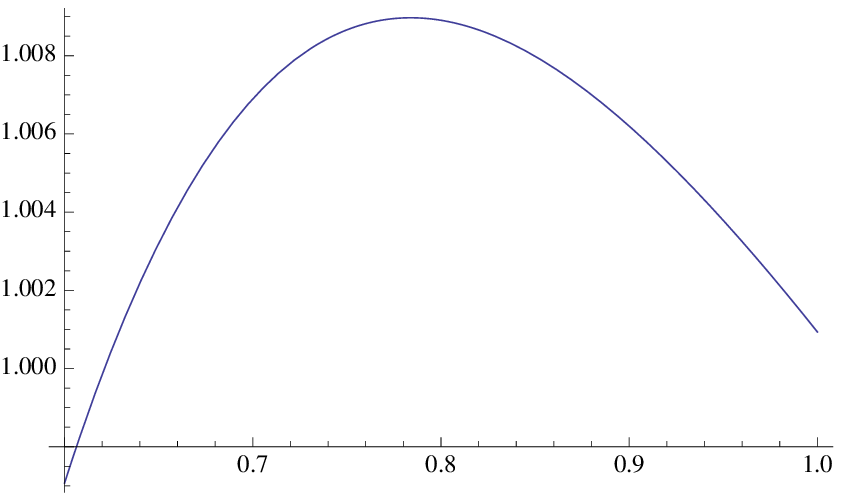}
(c)\includegraphics[width=0.3\columnwidth]{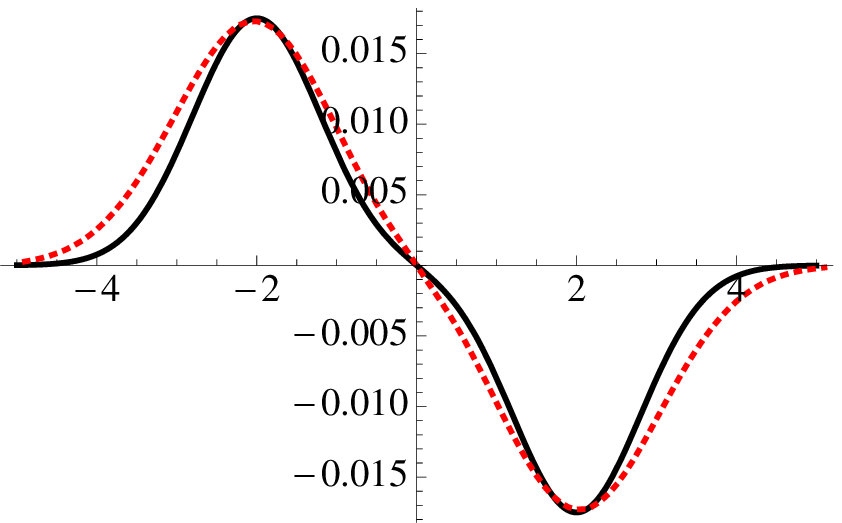} 
\par\end{centering}

\caption{\label{fig_Nonlinear-optimization-1}Nonlinear optimization of the
propagator eigenfunctions. (a) value of $\hat{\lambda}_{2}$ depending
on $y_{2}$ and with fixed $s_{2}=0.8$. (b) value of $\hat{\lambda}_{2}$
depending on $s_{2}$ with fixed $y_{2}=1$, (c) Approximation to
$\phi_{2}$ with $y_{2}=1$, $s_{2}=0.8$ (black), Reference solution
for $\phi_{2}$ (red).}
\end{figure}

\subsection{Quartic potential}

As a second numerical example, we use the diffusion in a one-dimensional
quartic potential: 

\[
V(x)=3x^{4}-6x^{2}+3,
\]
which has two minimum positions at $x=\pm1$. We seek to estimate
the second dominant eigenvalue $\lambda_{2}(\tau)$ and the corresponding
time scale $t_{2}=-\frac{\tau}{\log\lambda_{2}(\tau)}$ by applying
the Roothaan-Hall method above with Gaussian functions. We will then
compare the results to those obtained from a Markov state model discretization.
First, we generate a sample trajectory of the process as in Eq. (\ref{eq:smoluchowski_time_step}).
Here, we used a time step $\tau=10^{-3}$ and a total number of steps
$N=10^{7}$. From this sample, we computed an estimate $\hat{\mu}$
of the stationary density. We then computed the Markov state model
transition matrix and its eigenvalues, using a discretization of the
state space into $100$ sets.

For the application of the Roothaan-Hall method, we picked thirteen
Gaussian functions $\hat{\phi}_{i}$ with centres at

\[
x=-2,-1.5,-1.2,-1,-0.8,-0.5,0,0.5,0.8,1,1.2,1.5,2.
\]

The variances were set to $1$ for the functions centred at $x=-2,-1.5,0,1.5,2$,
and to $0.5$ for all others. Those functions were used as half-weighted
basis functions, meaning that we computed the entries of the $\mathbf{H}$-matrix
according to

\[
h_{ij}=\frac{1}{N-m}\sum_{k=1}^{N-m}\hat{\mu}^{-1/2}(x_{k})\hat{\phi}_{i}(x_{k})\hat{\mu}^{-1/2}(x_{k+m})\hat{\phi}_{j}(x_{k+m}),
\]
where $m$ is an integer corresponding to the lag time $m\tau$. We
similarly estimated the $\mathbf{S}$-matrix and then solved the generalized
eigenvalue problem Eq. (\ref{eq:roothan_hall_gen_ev}).

\begin{figure}
a)\includegraphics[width=0.4\textwidth,height=0.2\textheight]{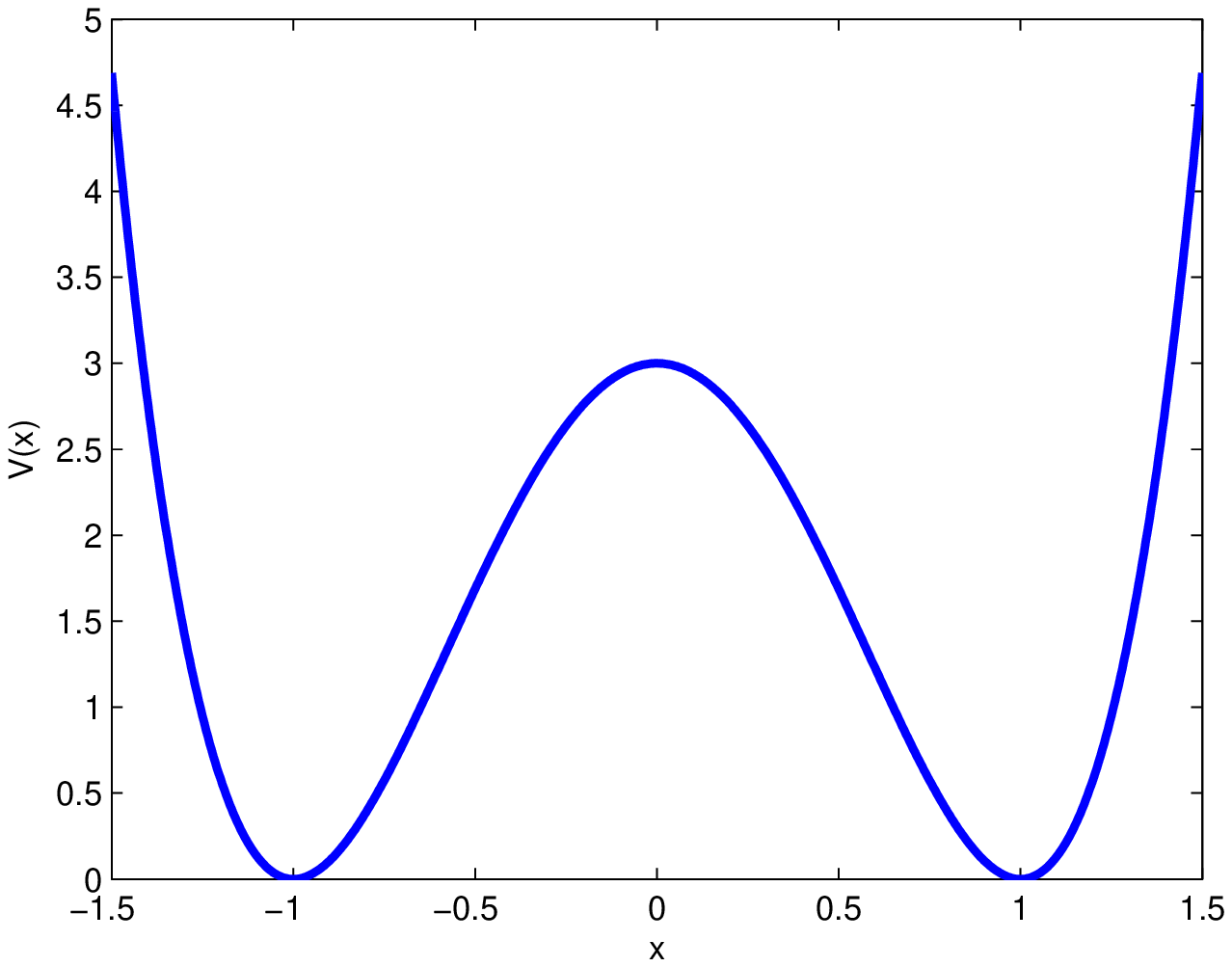}b)\includegraphics[width=0.4\textwidth,height=0.2\textheight]{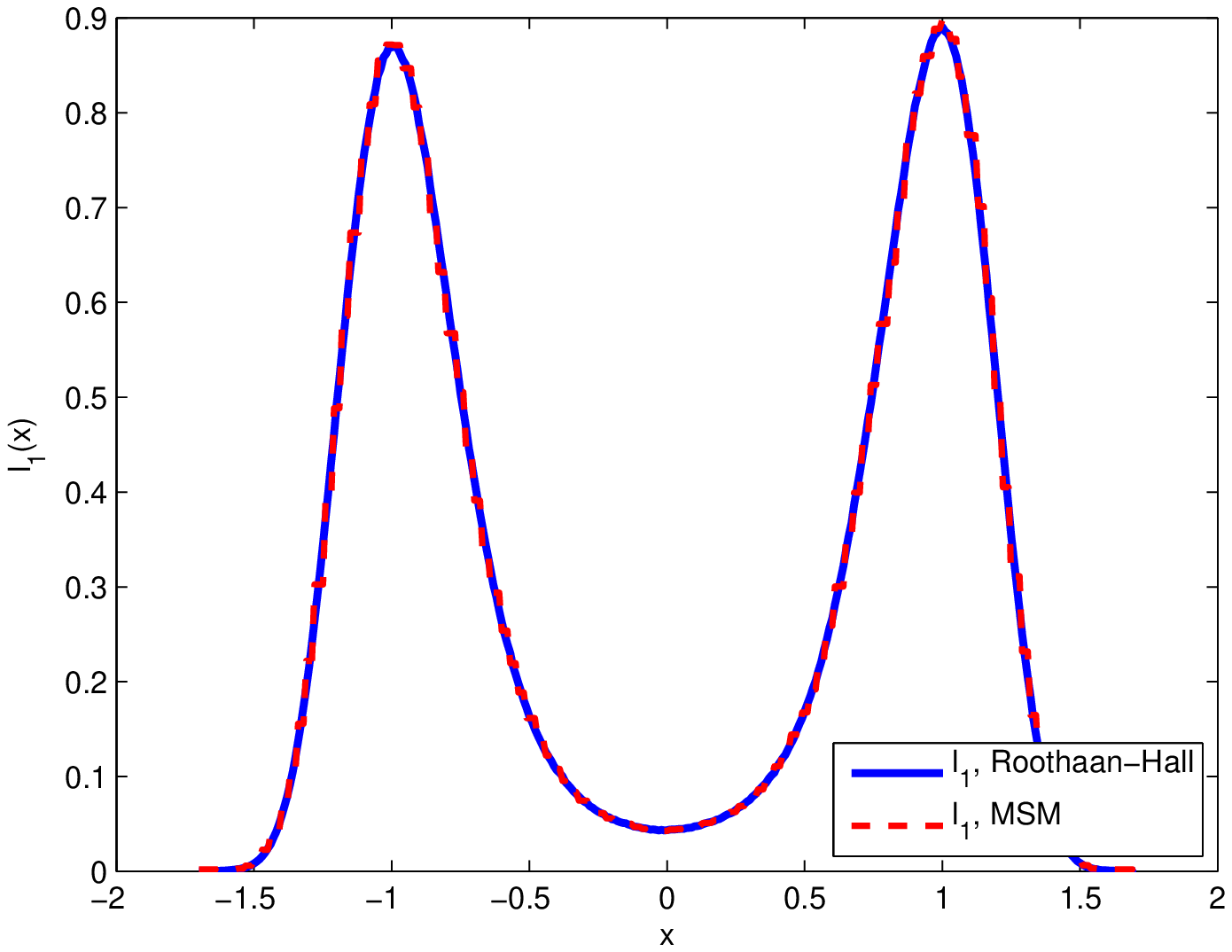}

c)\includegraphics[width=0.4\textwidth,height=0.2\textheight]{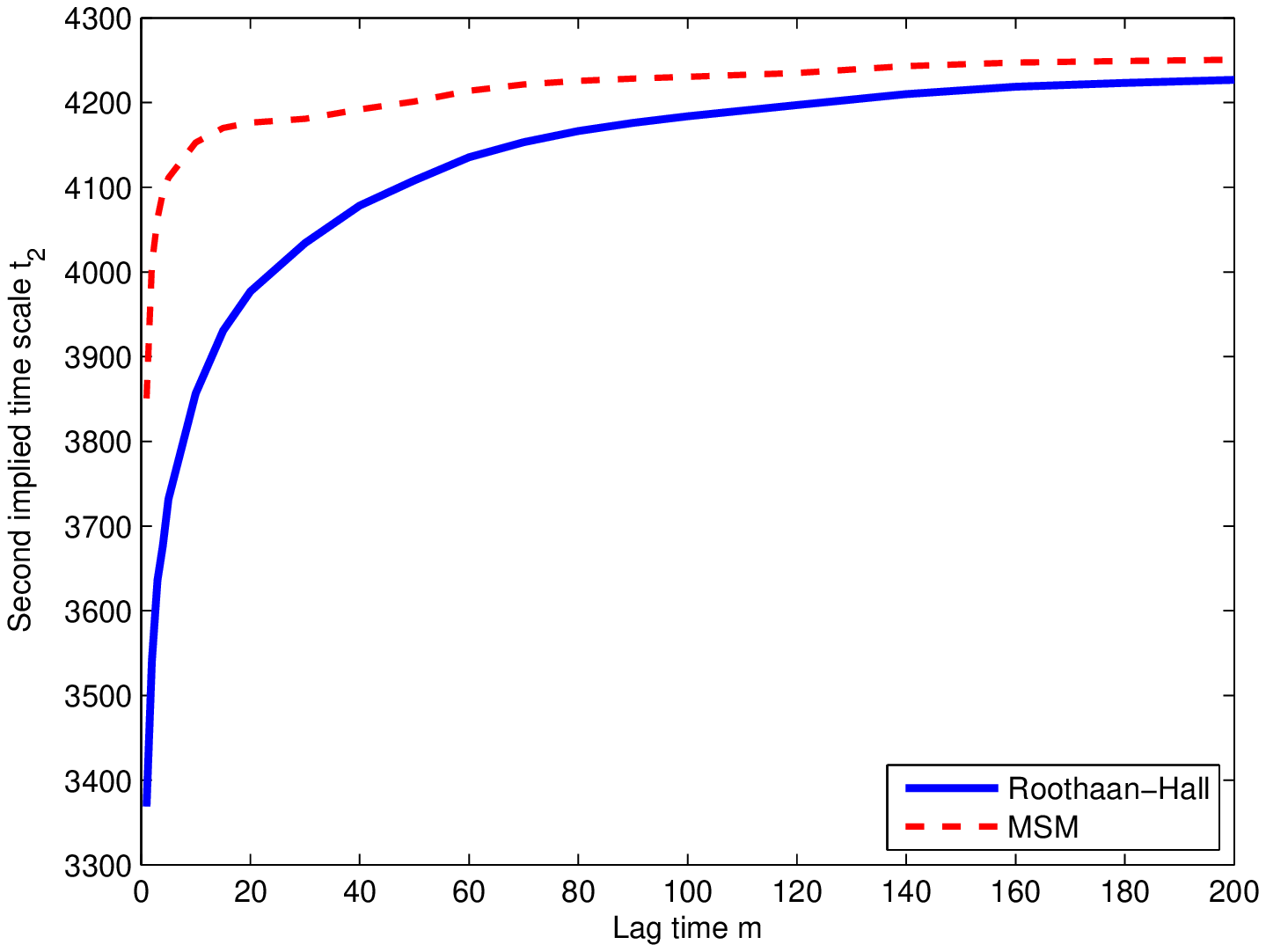}d)\includegraphics[width=0.4\textwidth,height=0.2\textheight]{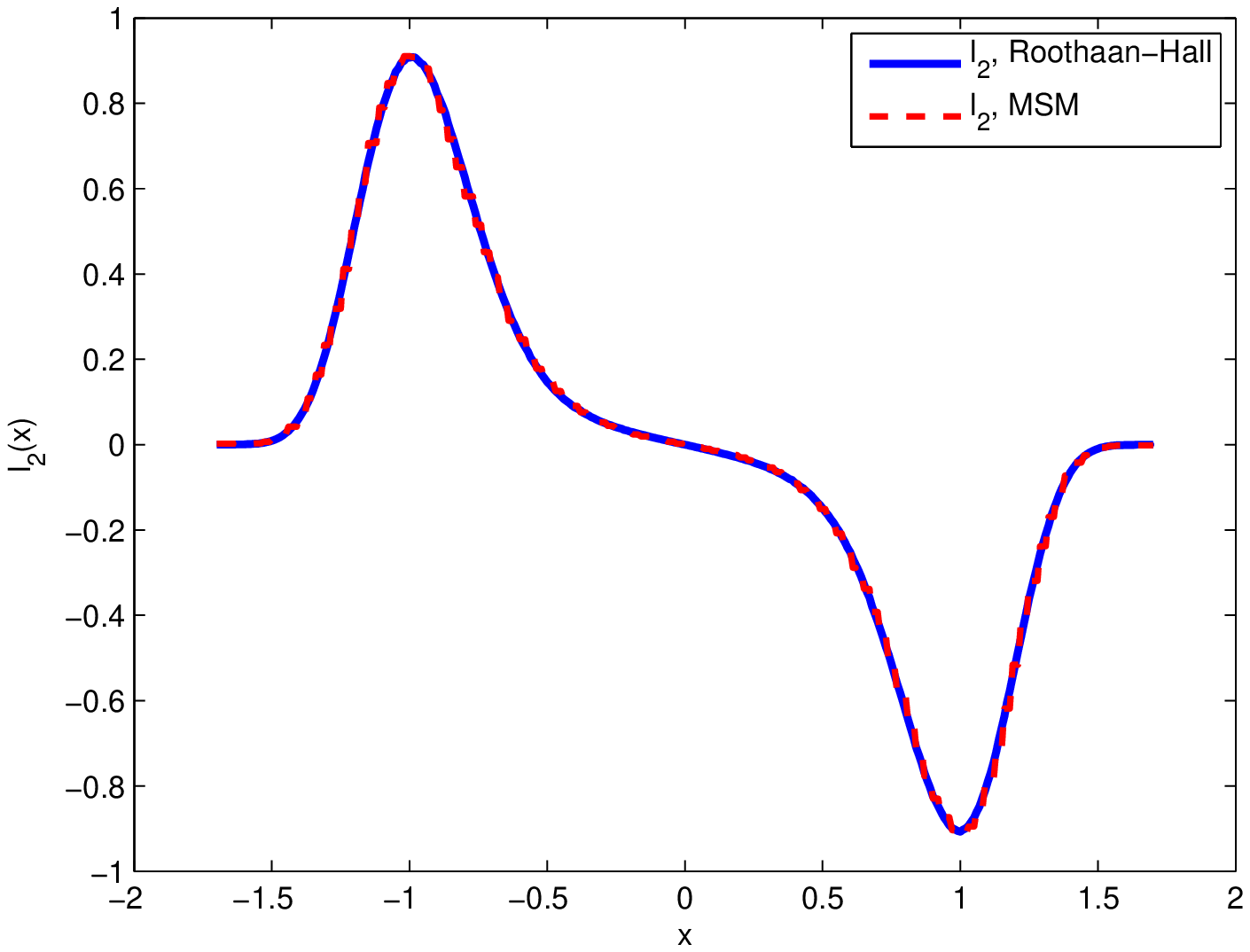}

\caption{Application of the Roothaan-Hall method with Gaussian basis functions
to a one-dimensional diffusion process, compared to a $100$-set MSM
discretization computed with the EMMA package \cite{SenneSchuetteNoe_JCTC12_EMMA1.2}.
a) The potential function $V$. b) Estimated stationary density $\hat{l}_{1}$
compared to the exact solution. c) Comparison of the second largest
eigenvalue $\lambda_{2}(\tau)$, estimated by the Roothaan-Hall method
and an MSM, both plotted against the lag time in integer multiples
of $\Delta t$.implied time scale $t_{2}$, with the lag time given
in integer multiples $m$ of the simulation step $\tau$. d) The second
eigenfunction $\hat{l}_{2}$ as estimated from both methods.\label{fig:Roothan-Hall_1D}}
\end{figure}

The results displayed in Figure \ref{fig:Roothan-Hall_1D} show that
we not only get a good approximation of both the first and second
eigenfunction in terms of smooth functions, but most importantly,
the second largest eigenvalue $\lambda_{2}(\tau)$ and the corresponding
time scale $t_{2}$ can be estimated comparably well with both methods.
While $100$ sets were used for the MSM discretization, only thirteen
basis function were used for the Roothaan-Hall method.

\section{Conclusions and outlook}

\label{sec_Conclusion}

Here, we have formulated a variational principle for Markov processes
where the dominant eigenfunctions are approximated by maximizing a
Rayleigh coefficient, which - in the limit of the exact eigenfunctions
- is identical to the true eigenvalues. This is the formulation needed
to attack the problem of estimating the slow processes in stochastic
dynamical systems with a much wider methodology than by the presently
used class of Markov State Models. In particular, the entire toolbox
of quantum mechanics where many decades of research have gone into
the development of eigenfunction approximation methods for high-dimensional
systems becomes available.

From a practical point of view, a main achievement of the present
study is that the Rayleigh coefficient can be estimated from simulation
data as it is equivalent to an autocorrelation function of the appropriately
weighted test function. The autocorrelation estimates are such that
they can be fed by many short simulations distributed across state
space and do not require the direct simulation of the slow processes
in a single long trajectory. This is an important advantage in dealing
with the sampling problem that arises in simulating metastable dynamical
systems.

A main use of the present approach will be to facilitate the development
of adaptive discretization algorithms of high-dimensional state spaces
for the computational characterization of complex dynamical processes.
The Rayleigh coefficient derived here represents a practically accessible
and theoretically solid functional to guide such an adaptive discretization
algorithm. In contrast to Markov state models, such an approximation
approach does not necessarily need to use the same basis set for all
eigenfunctions. Especially for reversible dynamics, different eigenfunctions
can be approximated separately, thus possibly permitting the use of
relatively small basis sets.

For a given class of dynamical systems, a basis sets must be selected
that is appropriate to model the regularity of the solution. For high-dimensional
processes such as molecular dynamics, Gaussian basis functions might
be a workable solution since they be well combined with clustering-based
identification of center positions and permit the analytical calculation
of some quantities such as the overlap integral. An interesting alternative
approach is to build the Basis set upon weakly coupled subsets of
internal molecular coordinates, as suggested in the mean field approach
developed in Ref. \cite{Friesecke_MMS09_MeanField}. The usefulness
of these and other approaches for complex molecular systems will be
investigated in future studies. Furthermore, subsequent studies will
deal with the error caused by the projection on a finite-dimensional
subspace depending on the choice of basis functions, as well as with
statistical considerations, such as the efficient evaluation of uncertainties
of the estimated Rayleigh coefficients.

\section*{Acknowledgements}

We gratefully acknowledge the DFG research center \textsc{Matheon}
and the Berlin Mathematical School (BMS) for funding. We would also
like to thank a number of colleagues for stimulating discussions,
in particular: Christof Schütte, Guillermo Perez-Hernandez, Bettina
Keller, Jan-Hendrik Prinz, Benjamin Trendelkamp-Schroer (all FU Berlin),
and John D. Chodera (UC Berkeley).

\bibliographystyle{plain}
\bibliography{qm-approach,all,/Users/noe/data/my_papers/bib/own}

\end{document}